\documentclass[12pt]{article}

\usepackage{amsmath}
\numberwithin{equation}{section}
\usepackage{amsfonts}
\usepackage{amsthm}
\usepackage{graphicx}
\usepackage[utf8]{inputenc}
\inputencoding{latin1}
\usepackage[usenames,dvipsnames]{color}
\usepackage[table]{xcolor}
\usepackage{multirow}
\usepackage{array}
\usepackage{float}
\floatstyle{ruled}
\newfloat{algorithm}{tbp}{loa}
\floatname{algorithm}{Algorithm}
\usepackage{caption}
\usepackage{subcaption}

\renewcommand\epsilon\varepsilon

\newtheorem{remark}{Remark}[section]
\newtheorem{theorem}[remark] {Theorem}

\newcounter{noqed}
\renewcommand{\qed}{\ifmmode\text{ }\fi\rule[-.05em]{.3em}{.7em}
\setcounter{noqed}{0}}
\renewenvironment{proof}[1][{}]{
  \noindent{\bf Proof#1.}\setcounter{noqed}{1}
}{
  \ifnum\value{noqed}=1\qed\fi\par\medskip
} 

\renewcommand\theta\vartheta

\begin{document}

%%%\begin{frontmatter}
\title{A Network Model characterized by\\ a Latent Attribute Structure
  with Competition} 

\author{Paolo Boldi\footnote{Dipartimento di
      Informatica, Universit\`{a} degli Studi di Milano, Via Comelico
      39/41 - 20135 Milano, Italy}, Irene Crimaldi\footnote{IMT
    Institute for Advanced Studies Lucca, Piazza San Ponziano 6,
    I-55100 Lucca, Italy}, Corrado Monti\footnote{Dipartimento di
      Informatica, Universit\`{a} degli Studi di Milano, Via Comelico
      39/41 - 20135 Milano, Italy}}

\maketitle

\begin{abstract}
The quest for a model that is able to explain, describe, analyze and
simulate real-world complex networks is of uttermost practical, as well
as theo\-re\-ti\-cal, interest.  In this paper we introduce and study a
network model that is based on a latent attribute structure: each node
is characterized by a number of features and the probability of the
existence of an edge between two nodes depends on the features they
share. Features are chosen according to a process of Indian-Buffet
type but with an additional random ``fitness'' parameter attached to
each node, that determines its ability to transmit its own features to
other nodes.  As a consequence, a node's connectivity does not depend on
its age alone, so also ``young'' nodes are able to compete and succeed 
in acquiring links.  One of the advantages of our model for the latent
bipartite ``node-attribute'' network is that it depends on few
parameters with a straightforward interpretation.  We provide some
theoretical, as well experimental, results regarding the power-law
behavior of the model and the estimation of the parameters. By
experimental data, we also show how the proposed model for the
attribute structure naturally captures most local and global
properties (e.g., degree distributions, connectivity and distance
distributions) real networks exhibit.\\

\noindent{\em keyword:}
Complex network, social network, attribute matrix, Indian Buffet process 
\end{abstract}

%%%%\end{frontmatter}

%\tableofcontents

\section{Introduction}

Complex networks are a unifying theme that emerged in the last decades
as one of the most important topics in many areas of science; the
starting point is the observation that many networks arising from
different types of interactions (e.g., in biology, physics, chemistry,
economics, technology, on-line social activity) exhibit surprising
similarities that are partly still unexplained. The quest for a model
that is able to explain, describe, analyze and simulate those
real-world complex networks is of uttermost practical, as well as
theoretical, interest.

The classical probabilistic model of graphs by Erd\H{o}s and
R\'enyi~\cite{ErdosRenyi} soon revealed itself unfit to describe
complex networks because, for example, it fails to produce a power-law
degree distribution. One of the first attempts to try to obtain more
realistic models was~\cite{BarabasiAlbert}, where the idea of
\emph{preferential attachment} was first introduced: nodes tend to
attach themselves more easily to other nodes that are already very
popular, i.e. with an high number of links. Similar models were
proposed by~\cite{AielloEtal} and~\cite{KumarEtal}. The general
approach of these and other attempts is to produce probabilistic
frameworks (typically with one or more parameters) giving rise to
networks with statistical properties that are compatible with the ones
that are observed in real-world graphs: degree distribution is just
one example; other properties are degree-degree correlation,
clustering coefficients, distance distribution, 
etc.~\cite{general-sna-book}.

The task of modeling the network is often undertaken
directly~\cite{BarabasiAlbert,KumarModel}, but recently some authors
proposed to split it into two steps (see,
e.g.,~\cite{latapy_bipartite,lattanzi_affiliation}). This proposal
stems from the observation that many complex networks contain two
types of entities: actors on one hand, and groups (or features) on the
other; every actor belongs to one, or more, groups (or can exhibit
one, or more, features), and the common membership to groups (or the
sharing of features) determines a relation between actors. The idea of
an underlying {\em bipartite network} such that interpersonal
connections follow from inter-group connections, derives from
sociology; a seminal paper presented by Breiger \cite{breiger} in 1974
described this dualism between ``persons and groups''. This idea has
been proved precious in social networks and their mathematical
modelization \cite{lattanzi_affiliation}.

In particular, many authors \cite{MGJ2009} distinguish between two
kinds of models: class-based models -- such as \cite{blockmodel} --
assume that every node belongs to a single class, while feature-based
models use many features to describe each node.  A well-known
shortcoming of the first is the proliferation of classes, since
dividing a class according to a new feature leads to two different
classes.  To overcome this limitation, classical class-based models
have been extended to allow mixed membership, like
in~\cite{mixed-membership}.  Feature-based models naturally assume
this possibility.  Within them, some authors (such as
\cite{latent-factor-models}) propose real-valued vectors to associate
features to nodes; others instead assume only binary features, in
which a node either exhibits a feature or it does not (see
e.g.~\cite{MGJ2009}). This assumption is simple and natural, and it
significantly simplifies the analysis of the model.

A natural model for the evolution of such binary bipartite graphs
comes from Bayesian statistics and it is known as the Indian Buffet
process, introduced by Griffiths \emph{et al.} \cite{gha, GG06, GG11}
and, subsequently extended and studied by many authors \cite{bcpr-ibm,
  BJP, TG, TJ}.  The process defines a plausible way for features to
evolve, always according to a \emph{rich-get-richer} principle:
because of this, it represents a promising model for affiliation
networks.  Since the Indian Buffet process provides {\em a prior
  distribution} in Bayesian statistics, these models have been used to
reconstruct affiliation networks, with an unknown number of features, 
from data where only friendship relations between actors are
available. An important work in this direction is
\cite{MGJ2009}. However, the standard Indian Buffet process has a
drawback as a model for real networks: the exchangeability assumption
is often untenable in applications.

In this paper we propose and analyze a model that combines two features
cha\-rac\-te\-ri\-zing the evolution of a network:
\begin{enumerate}
 \item Behind the adjacency matrix of a network there is a {\em latent
   attribute structure} of the nodes, in the sense that each node is
   characterized by a number of features and the probability of the
   existence of an edge between two nodes depends on the features they
   share. In other words, the adjacency matrix of a network hides a
   bipartite network describing the attributes of the nodes.

\item Not all nodes are equally successful in transmitting their own
  attributes to the new nodes. Each node $n$ is characterized by a
  {\em random fitness parameter} $R_n$ describing its ability to transmit
  the node's attributes: the greater the value of the random variable
  $R_n$, the greater the probability that a feature of $n$ will also
  be a feature of a new node, and so the greater the probability of
  the creation of an edge between $n$ and the new node.  Consequently,
  a node's connectivity does not depend on its age alone (so also
  ``young'' nodes are able to compete and succeed in acquiring
  links). We refer to this aspect as {\em competition}. 
 \end{enumerate}

We shape the first aspect by the definition of a model which connects
the pair of attribute-vectors of two nodes, say $i$ and $j$, to the
probability of the existence of an edge between $i$ and $j$. Other
examples, which are related to the Bayesian framework based on the
standard Indian Buffet model, can be found in \cite{MGJ2009, msh,
  palla-et-al, SCJ}.

We model the second aspect by the definition of a stochastic dynamics
for a bipartite ``node-attribute'' network, where the probability that
a new node exhibits a certain attribute depends on the ability,
represented by some random fitness parameter, of the previous nodes
possessing that attribute in transmitting it. It is worthwhile to
underline that in our model, as in the standard Indian Buffet process,
the collection of attributes is potentially unbounded. Thus, we do not
need to specify a maximum number of latent attributes {\it a priori}.

We were inspired by the recent generalization of the Indian buffet
process presented in \cite{bcpr-ibm}. However, the model presented
here is in some sense simpler since the parameters (that will be
introduced and analyzed in the next sections) play a role that is clearer
and more intuitive. Specifically, we have two parameters ($\alpha$
and $\beta$) that control the number of new attributes each new node
exhibits (in particular $\beta>0$ tunes the power-law behaviour of the
overall number of different observed attributes), whereas the random
fitness parameters $R_i$ impact on the probability of the new nodes to
inherit the attributes of the previous nodes. With respect to the
model in \cite{bcpr-ibm}, we lose some mathematical properties, but we
will show that some important results still hold true and they allow
us to estimate the parameters and, in particular, the exponent of the
power-law behavior.

Regarding the use of fitness parameters, we recall the work by
Bianconi and Barab\'asi \cite{BB} that introduced some fitness
parameters describing the ability of the nodes to compete for links.
The difference between their model and ours consists in the fact that
in \cite{BB} the fitness parameters appear explicitly in the
edge-probabilities; while in our model they affect the evolution of
the attribute matrix and then play an implicit role in the evolution
of the connections.

Summing up, the present work have different aims: firstly, we propose
a simple model for the latent bipartite ``node-attribute'' network,
where the role played by the single parameters is straightforward and
easy to be interpreted; secondly, dif\-fe\-ren\-tly from other network
models based on the standard Indian Buffet process, we take into
account the aspect of competition and, like in \cite{bcpr-ibm}, we
introduce random fitness parameters so that nodes have a different
relevance in transmitting their features to the next nodes; finally,
we provide some theoretical, as well experimental, results regarding
the power-law behavior of the model and the estimation of the
parameters. By experimental data, we will also show how the proposed
attribute structure naturally leads to a complex network model.

The paper is structured as follows. In Section~\ref{sec-model}, we
introduce a model for the evolution of the attribute matrix and we
provide theoretical results and tools regarding the estimation and the
analysis of the quantities characterizing the model. These methods are
then tested by simulations in Section~\ref{sec:sim-Z}. In order to
produce a graph out of the attribute structure, in
Section~\ref{sec:graph} we illustrate different models for the
edge-probabilities that are based on the attribute matrix. The
properties of the generated graphs are studied by simulation in
Section~\ref{sec:graph-simulations}. Finally, in
Section~\ref{real-data} we analyze a real dataset and, then, in
Section~\ref{conclusions} we sum up the main novelties and merits of
our work and we illustrate some possible future lines of research.

\section{A model for the evolution of the attribute matrix}
\label{sec-model}

We assume that the nodes enter the network sequentially so that node
$i$ represents the node that comes into the network at time $i$. Let
$\mathcal X$ be an unbounded collection of possible attributes that a
node can exhibit. (This means that we do not specify the total number
of possible attributes {\it a priori}.)  Each node is assumed to have
only a finite number of attributes and different nodes can share one
or more attributes.

Let $Z$ be a binary bipartite network where each row $Z_n$ represents
the attributes of the node $n$: $Z_{n,k}=1$ if node $n$ has attribute
$k$, $Z_{n,k}=0$ otherwise. We assume that each $Z_n$ remains
unchanged in time, in the sense that every node decides its own
features (attributes) when it arrives and then it will never change them
thereafter. This assumption is quite natural in many contexts, e.g. 
in genetics.

\smallskip
In all the sequel we postulate that $Z$ is left-ordered. This means
that in the first row the columns for which $Z_{1,k}=1$ are grouped on
the left and so, if the first node has $N_1$ features, then the
columns of $Z$ with index $k\in\{1,\dots, N_1\}$ represent these
features. The second node could have some features in common with the
first node (those corresponding to indices $k$ such that $k=1,\dots,
N_1$ and $Z_{2,k}=1$) and some, say $N_2$, new features. The latter are
grouped on the right of the sets for which $Z_{1,k}=1$, i.e., the
columns of $Z$ with index $k\in\{N_1+1,\dots, N_2\}$ represent the new
features brought by the second node. This grouping structure persists
throughout the matrix $Z$.

\smallskip
Here is an example of a $Z$ matrix with $n=4$ nodes; in gray we show
the new features adopted by each node ($N_1=3$, $N_2=2$, $N_3=3$,
$N_4=2$ in this example); observe that, for every node $i$, the $i$-th
row contains 1's for all the columns with indices $k \in
\{N_1+\dots+N_{i-1}+1,\dots,N_1+\dots+N_i\}$ (they represent the new
features brought by $i$); moreover some elements of the columns with
indices $k \in \{1,\dots,N_1+\dots+N_{i-1}\}$ are also $1$'s (features
brought by previous nodes and that also node $i$ decided to adopt):
\[
	\left(
	\begin{array}{cccccccccccc}
	\cellcolor{gray!20}1 & \cellcolor{gray!20}1 & \cellcolor{gray!20}1 
& 0 & 0 & 0
	& 0 & 0 & 0 & 0 \\
	1 & 0 & 1 & \cellcolor{gray!20}1 & \cellcolor{gray!20}1 & 0 & 0 & 0 
& 0 & 0 \\
	0 & 1 & 1 & 1 & 0 & \cellcolor{gray!20}1 & \cellcolor{gray!20}1 &
	\cellcolor{gray!20}1 & 0 & 0 \\
	1 & 1 & 1 & 0 & 1 & 1 & 0 & 1 & \cellcolor{gray!20}1 & 
\cellcolor{gray!20}1 \\
	\end{array}
	\right)
\]

\smallskip
We will describe the dynamics using a culinary metaphor (similarly to
what some authors do for other models, see Chinese Restaurant
\cite{PIT}, Indian Buffet process \cite{GG06, GG11, TG} and their
generalizations \cite{BCL, bcpr-ibm}). We identify the nodes with the
customers of a restaurant and the attributes with the dishes, so that
the dishes tried by a customer represent the attributes that a node
exhibits.\\

\medskip
Fix $\alpha>0$ and $\beta\in (-\infty,1]$.  Also, let Poi$(\lambda)$
  denote the Poisson distribution with mean $\lambda> 0$.  Customer
  (node) $n$ is attached a random weight (that we call, in accordance
  with the usage in Network Theory, {\em fitness parameter}) $R_n$. We
  assume that each $R_n$ is independent of $R_1,\dots,R_{n-1}$ and of
  the dishes (features) experimented by customers $1,\dots, n$. The
  fitness parameter $R_n$ affects the choices of the future customers
  (those after $n$), while the choices of customer $n$ are affected by
  the fitness parameters and the choices of the previous ones. Indeed,
  it may be the case that different customers have different
  relevance, due to some random cause, that does not affect their
  choices but is relevant to the choices of future customers (i.e.,
  their capacity of being followed).
  
  The dynamics is as follows.  Customer (node) 1 tries $N_1$ dishes
  (features), where $N_1$ is Poi$(\alpha)$-distributed. For each
  $n\geq 1$, let $S_n$ be the collection of dishes (features)
  experimented by the first $n$ customers (nodes).  For the customers
  following the first one, we have that:
\begin{itemize}
\item Customer $n+1$ selects a subset $S_n^*$ of $S_n$. Each
  $k\in S_n$ is included or not into $S_n^*$ independently of the
  other members of $S_n$. The inclusion probability is
\begin{equation}
P_n(k)=\frac{\sum_{i=1}^n R_i Z_{i,k}}{\sum_{i=1}^n R_i}\,,
\label{eq:define_inclusion_prob}
\end{equation} 
where $Z_{i,k}=1$ if $\{$customer $i$ has selected dish $k\}$ and
$Z_{i,k}=0$ otherwise. It is a preferential attachment rule: the
larger the weight of a dish $k$ at time $n$ (given by the numerator of
(\ref{eq:define_inclusion_prob}), i.e., the total value of the random
variables $R_i$ associated to the customers that have chosen it until
time $n$), the greater the probability that it will be chosen by the
future customer $n+1$.
\footnote{As we will discuss in the conclusions
  (Sec.~\ref{conclusions}), we can generalize our model by introducing
  another parameter $c\geq 0$ in the inclusion probabilities so that
\begin{equation*}
P_n(k)=\frac{\sum_{i=1}^n R_i Z_{i,k}}{c+\sum_{i=1}^n R_i}\,.
\end{equation*} 
For the moment, we set $c=0$. Note that this choice implies $P_n(k)=1$
for all $n$ and $k=1,\dots, N_1$. Therefore, we could consider the
first node and its features ``fictitious'', in the sense that the
``true'' dynamics is for $n\geq 2$ and $k\geq N_1+1$.  }
\item In addition to $S_n^*$, customer $n+1$ also tries $N_{n+1}$
  new dishes, where $N_{n+1}$ is Poi$(\Lambda_n)$-distributed with
\begin{equation}
\Lambda_n=\frac{\alpha}{\left(\sum_{i=1}^n R_i\right)^{1-\beta}}.
\label{eq:define_lambda}
\end{equation}
\end{itemize}
For each $k$ in $S_{n+1}$, the matrix element $Z_{n+1,k}$ is set equal
to $1$ if customer $n+1$ has selected dish $k$, equal to zero
otherwise.

Besides the assumption of independence, we also assume that the random
parameters $R_n$ are identically distributed with $R_n\geq v$ for
each $n$ and a certain number $v>0$, and $E[R_n^2]<+\infty$.

We set $E[R_n]=m_R$ and $L_n=card(S_n)=\sum_{i=1}^n N_i$, i.e.
\begin{equation*}
\begin{split}
L_n&=\hbox{overall number of different dishes experimented by the first }
n \hbox{ customers}\\
&=\hbox{overall number of different observed attributes for the first } n 
\hbox{ nodes}.
\end{split}
\end{equation*}
In the previous example, we have $L_1=3,L_2=5,L_3=8,L_4=10$.

The meaning of the parameters is the following. The random fitness
parameters $R_n$ fundamentally control the probability of transmitting the
attributes to the new nodes. 
%%% All the parameters (the random fitness
%%% parameter $R_n$, $\alpha$ and $\beta$) together control the total
%%% number of different tried dishes (total number of different observed
%%% attributes) and the matrix $Z$. In particular, 
The main effect of $\beta$ is that it regulates the asymptotic
behavior of the random variable $L_n$ (see Theorem \ref{th-L}). In
particular, $\beta>0$ is the power-law exponent of $L_n$.  The main
effect of $\alpha$ is the following: the larger $\alpha$, the larger
the total number of new tried dishes by a customer (and so the larger
the total number of $1$'s in a row of the binary matrix $Z$). It is
worth to note that $\beta$ fits the asymptotic behaviour of $L_n$
(in particular, the power-law exponent of $L_n$) and, separately, 
$\alpha$ fits the number of observed features.

The mathematical formalization of the above model can be performed by means
of random measures~\cite{KING} with atoms corresponding to the tried
dishes (observed attributes), similarly to~\cite{bcpr-ibm, BJP,
  TJ}. More precisely, besides the sequence of positive real random
variables $(R_n)$, we can define a sequence of random measures
$(M_n)$, such that each $M_{n+1}$ is, conditionally on the past
$(M_i,R_i: i\leq n)$, a Bernoulli random measure with a hazard measure
$\nu_n$, having a discrete part related to the points $k$ in $S_n$ and
their weights $P_n(k)$ and a diffuse part with total mass equal to
$\Lambda_n$.

\subsection{Theoretical results regarding the estimation of the parameters 
$\alpha$ and~$\beta$}

In this section we prove some properties regarding the asymptotic
behavior of $L_n$. In particular, the first result shows a
logarithmic behavior for $\beta=0$ and a power-law behavior for
$\beta\in (0,1]$. These results allow us to define suitable estimators
  for $\beta$ and $\alpha$.

\begin{theorem}\label{th-L}
Using the previous notation, the
following statements hold true:
\begin{itemize}
\item{a)} $\sup_n L_n=L<+\infty$ a.s. for $\beta<0$;
\item{b)} ${L_n}/{\ln(n)}\stackrel{a.s.}\longrightarrow 
\alpha/m_R$ for $\beta =0$;
\item{c)} ${L_n}/{n^{\beta}}\stackrel{a.s.}\longrightarrow
{\alpha}/{(\beta\, m_R^{1-\beta})}$ for $\beta\in (0,1]$.
\end{itemize}
\end{theorem}

\begin{proof}  
Let us prove assertion a), first.  Let ${\cal F}_i$ be the natural
$\sigma$-field associated to the model until time $i$ and set
$\Lambda_0=\alpha$. Since, conditionally on $\mathcal{F}_i$, the
distribution of $N_{i+1}$ is Poi$(\Lambda_i)$, we have
\begin{gather*}
P(N_{i+1}\geq 1)=E\bigl[P(N_{i+1}\geq
1\mid\mathcal{F}_i)\bigr]
\leq E[\Lambda_i].
\end{gather*}
Since $R_i\geq v>0$, we obtain
$$
\sum_i P(N_{i+1}\geq 1)\leq \alpha \sum_i \frac{1}{(vi)^{(1-\beta)}}<+\infty
$$ (where the convergence of the series is due to the assumption
$\beta<0$).  By the Borel-Cantelli lemma, we conclude that
$$
P\bigl(N_i>0 \hbox{ infinitely often }\bigr)
=P\bigl(N_i\geq 1 \hbox{ infinitely often}\bigr)
=0.
$$
Hence, if $\beta<0$, there is a random index $N$ such that $L_n=L_N$ a.s. 
for all $n\geq N$, which concludes the proof of a). 

\indent The assertion c) is trivial for $\beta=1$ since, in
this case, $L_n$ is the sum of $n$ independent random variables with
distribution ${\mathcal P}(\alpha)$ and so, by the classical strong
law of large numbers, $L_n/n\stackrel{a.s.}\longrightarrow\alpha$. 

\indent Now, let us prove assertions b) and c) for $\beta\in
        [0,1)$. Define
\begin{gather*}
\lambda(\beta)=\frac{\alpha}{m_R}\,\text{ if }\,\beta=0\quad\text{and}
\quad\lambda(\beta)=\frac{\alpha}{\beta\,m_R^{1-\beta}}\,\,\text{ if }\,
\beta\in (0,1),
\\a_n(\beta)=\log{n}\,\text{ if }\,\beta=0\quad\text{and}\quad 
a_n(\beta)=n^\beta\,\text{ if }\,\beta\in (0,1).
\end{gather*}
We need to prove that
\begin{gather*}
\frac{L_n}{a_n(\beta)}\overset{a.s.}\longrightarrow\lambda(\beta).
\end{gather*}
First, we observe that we can write 
\begin{gather*}
\frac{\sum_{i=1}^{n-1}\Lambda_i}{a_n(\beta)}
=
\alpha\frac{\sum_{i=1}^{n-1}i^{\beta-1}\,\bigl(\overline{R}_i\bigr)^{\beta-1}}
{a_n(\beta)},
\end{gather*}
where, by the strong law of the large numbers, 
$$
\overline{R}_i=\frac{\sum_{j=1}^i R_j}{i} \stackrel{a.s}\longrightarrow m_R.
$$ Therefore, since $\sum_{i=1}^{n-1}i^{\beta-1}/a_n(\beta)$ converges
to $1$ when $\beta=0$ and to $1/\beta$ when $\beta\in (0,1)$, we get
\begin{equation}\label{lim0}
\frac{\sum_{i=1}^{n-1}\Lambda_i}{a_n(\beta)}\overset{a.s.}
\longrightarrow\lambda(\beta). 
\end{equation} 
Next, let 
%%%  
%%% ${\cal F}_i$ be the natural $\sigma$-field associated to
%%% the model until time $i$ and
%%% 
us define
\begin{gather*}
T_0=0\quad\text{and}\quad 
T_n=\sum_{i=1}^n\frac{N_i-E[N_i\mid\mathcal{F}_{i-1}]}{a_i(\beta)}=
\sum_{i=1}^n\frac{N_i-\Lambda_{i-1}}{a_i(\beta)}.
\end{gather*}
Then, $(T_n)$ is a martingale with respect to $(\mathcal{F}_n)$ and
\begin{gather*}
E[T_n^2]=
\sum_{i=1}^n\frac{E\bigl[(N_i-\Lambda_{i-1})^2\bigr]}{a_i(\beta)^2}=
\sum_{i=1}^n\frac{E\bigl\{E\bigl[(N_i-\Lambda_{i-1})^2\mid\mathcal{F}_{i-1}\bigr]
\bigr\}}{a_i(\beta)^2}=\sum_{i=1}^n\frac{E[\Lambda_{i-1}]}{a_i(\beta)^2}.
\end{gather*}
Since $R_i\geq v>0$, it is easy to verify that 
$E[\Lambda_i]=\text{O}(i^{-(1-\beta)})$ and so 
$\sup_n E[T_n^2]=\sum_{i=1}^\infty\frac{E[\Lambda_{i-1}]}{a_i(\beta)^2}<\infty$. 
Thus, $(T_n)$ converges a.s., and the Kronecker's lemma implies
$$
\frac{1}{a_n(\beta)}\sum_{i=1}^n a_i(\beta)\frac{(N_i-\Lambda_{i-1})}{a_i(\beta)}
\stackrel{a.s.}\longrightarrow 0,
$$
so finally 
\begin{gather}\label{lim}
\lim_n\frac{L_n}{a_n(\beta)}=
\lim_n\frac{\sum_{i=1}^n N_i}{a_n(\beta)}=
\lim_n\frac{\sum_{i=1}^n\Lambda_{i-1}}{a_n(\beta)}=
\lim_n\frac{\Lambda_0+\sum_{i=1}^{n-1}\Lambda_i}{a_n(\beta)}=
\lambda(\beta)\quad\text{a.s.}
\end{gather}
\end{proof}

The above result entails that $\ln(L_n)/ \ln(n)$ is a strongly
consistent estimator of $\beta\in[0,1]$. In fact:
\begin{itemize}
  \item if $\beta=0$ then
    $L_n\stackrel{a.s.}\sim\frac{\alpha}{m_R}\ln(n)$ as $n\to
    +\infty$; hence $\ln(L_n)\stackrel{a.s.}\sim\ln(\alpha/m_R) +
    \ln(\ln(n))$, therefore
    $\ln(L_n)/\ln(n)\stackrel{a.s.}\sim\ln(\alpha/m_R)/\ln(n) +
    \ln(\ln(n))/\ln(n)\stackrel{a.s.}\to 0=\beta$;
  \item if $\beta>0$, we have $L_n\stackrel{a.s.}\sim \lambda(\beta)
    n^{\beta}$ as $n\to +\infty$ so
    $\ln(L_n)\stackrel{a.s.}\sim\ln(\lambda(\beta)) + \beta \ln(n)$,
    hence
    $\ln(L_n)/\ln(n)\stackrel{a.s.}\sim\ln(\lambda(\beta))/\ln(n) +
    \beta\stackrel{a.s.}\to \beta$.
\end{itemize} 

\begin{remark}\label{remark-estimator-beta} 
\rm In practice, the value of $\ln(L_n)/\ln(n)$ may be quite far from
the limit value $\beta$ when $n$ is small.  Hence, it may be worth
trying to fit the power-law dependence of $L_n$ as a function of $n$
with standard techniques~\cite{CSNPLDED} and use the slope
$\widehat\beta_n$ of the regression line in the log-log plot as an
effective estimator for $\beta$.
\end{remark}

Finally, assuming that $\beta\in [0,1]$ and $m_R$ are known, 
we can get a strongly consistent estimator of
$\alpha$, as:
$$
m_R\,\frac{L_n}{\ln(n)}
\quad\hbox{for }\beta =0
\qquad\hbox{and}\qquad 
m_R^{1-\beta}\beta\,\frac{L_n}{n^{\beta}}
\quad\hbox{for } 0<\beta\leq 1.
$$ 

In practice, we assume $\beta$ equal to the estimated
value $\widehat{\beta}_n$ (as defined before) and we take $m_R$
equal to the estimated value $\overline{R}_n=\sum_{i=1}^n R_i/n$, if
the random parameters $R_i$ are known. In
Section~\ref{sec:monte-carlo-experiments}, we will discuss the case when
the random variables $R_i$ are unknown.

\begin{remark}\label{remark-estimator-alpha} 
\rm Once more, it may be better in practice to estimate
$\alpha$ as 
\begin{equation}\label{stima-alpha}
\begin{split}
\widehat\alpha_n&=m_R\,\widehat\gamma_n\qquad\hbox{when } \beta=0\\
\widehat\alpha_n&={\beta}\,m_R^{1-\beta}\,
\widehat\gamma_n \qquad\hbox{when } 0<\beta\leq 1,
\end{split}
\end{equation}
where $\widehat\gamma_n$ is the slope of the regression line in the
plot $\big(\ln(n), L_n\big)$ or in the plot $\big(n^{\beta}, L_n\big)$
according to whether $\beta=0$ or $\beta\in(0,1]$.
\end{remark}

We complete this section with a central-limit theorem that gives the
rate of convergence of $L_n/a_n(\beta)$ to $\lambda(\beta)$ when
$\beta\in [0,1]$.

\begin{theorem} 
If $\beta\in [0,1]$, then we have the following convergence in
distribution\footnote{Actually, the convergence is in the sense of
  the {\em stable} convergence, which is stronger than the convergence
  in distribution. Indeed, stable convergence is a form of convergence
  intermediate between convergence in distribution and convergence in
  probability.}:
\begin{gather*}
\sqrt{a_n(\beta)}\,\Bigl\{\frac{L_n}{a_n(\beta)}-\lambda(\beta)\Bigr\}
\stackrel{d}\longrightarrow\mathcal{N}\bigl(0,\,\lambda(\beta)\bigr).
\end{gather*}
\end{theorem}

\begin{proof} The result for $\beta=1$ follows from the classical central 
limit theorem, since, in this case, $L_n$ is the sum of $n$ independent
random variables with distribution ${\mathcal P}(\alpha)$.  Assume now
$\beta\in[0,1)$ and set $\Lambda_0=\alpha$. We first prove that
\begin{gather}\label{uje}
\sqrt{a_n(\beta)}\,\Bigl\{\frac{\sum_{i=1}^n\Lambda_{i-1}}{a_n(\beta)}-
\lambda(\beta)\Bigr\}\overset{P}\longrightarrow 0.
\end{gather}
By some calculations, condition \eqref{uje} is equivalent to
\begin{gather}\label{uje2}
\frac{\sum_{i=1}^{n-1}\,\big\{\bigl(\sum_{j=1}^i R_j\bigr)^{\beta-1}-
(m_R\,i)^{\beta-1}\bigr\}}{\sqrt{a_n(\beta)}}\overset{P}\longrightarrow 0.
\end{gather}
Since $R_j\geq v>0$, we have $m_R \geq v>0$ and we obtain
\begin{equation*}
\begin{split}
E\left[
\,\left|{(m_R\,i)^{\beta-1}-\left(\sum_{j=1}^i R_j\right)^{\beta-1}}\right|\,
\right]
&\leq
\frac{E\left[
\,\left|{\bigl(\sum_{j=1}^i R_j\bigr)^{1-\beta}-(m_R\,i)^{1-\beta}}\right|\,
\right]}
{(v\,i)^{2(1-\beta)}}
\\
&\leq
\frac{1}{(v\,i)^{2(1-\beta)}}\,
\frac{1-\beta}{(v\,i)^{\beta}}\,
E\left[\,\left|{\sum_{j=1}^i R_j-m_R\,i}\right|\,\right]\\
&=
\frac{1-\beta}{v^{2-\beta}}\,\frac{1}{i^{1-\beta}}
E\left[\,|{\overline{R}_i-m_R}|\,\right]\\
&\leq
\frac{1-\beta}{v^{2-\beta}}\,\frac{1}{i^{1-\beta}}
\sqrt{Var[\overline{R}_i]}
=\frac{(1-\beta)\sqrt{Var[R_1]}}{v^{2-\beta}}\,\frac{i^{\beta-1}}{\sqrt{i}}.
\end{split}
\end{equation*}

This proves condition \eqref{uje2} (and so \eqref{uje}). Indeed, we
have
\begin{equation*}
\begin{split}
&\frac{1}{\sqrt{a_n(\beta)}}
E\left[\,
\left|
\sum_{i=1}^{n-1}\,\left\{
\left(\sum_{j=1}^i R_j\right)^{\beta-1}-(m_R\,i)^{\beta-1}
\right\}
\right|
\,\right]
\leq\\
&\frac{1}{\sqrt{a_n(\beta)}}\sum_{i=1}^{n-1}
E\left[
\,
\left|
(m_R\,i)^{\beta-1}-\left(\sum_{j=1}^i R_j\right)^{\beta-1}
\right|\,
\right]
\leq\\
&\frac{(1-\beta)\sqrt{Var[R_1]}}{v^{2-\beta}}\,
\frac{1}{\sqrt{a_n(\beta)}}
\sum_{i=1}^{n-1}\frac{1}{i^{1-(\beta-1/2)}}
\to 0.
\end{split}
\end{equation*}

\smallskip Next, define
\begin{gather*}
T_n=\sqrt{a_n(\beta)}\,\Bigl\{\frac{L_n}{a_n(\beta)}-
\frac{\sum_{i=1}^n\Lambda_{i-1}}{a_n(\beta)}\Bigr\}=
\frac{\sum_{i=1}^n(N_i-\Lambda_{i-1})}{\sqrt{a_n(\beta)}}.
\end{gather*}

In view of \eqref{uje}, it suffices to show that
$T_n\stackrel{d}\longrightarrow\mathcal{N}\bigl(0,\,\lambda(\beta)\bigr)$.

To this end, for $n\geq 1$ and $i=1,\ldots,n$, define
\begin{gather*}
T_{n,i}=\frac{N_i-\Lambda_{i-1}}{\sqrt{a_n(\beta)}},\quad\mathcal{G}_{n,0}=
\mathcal{F}_0\quad\text{and}\quad\mathcal{G}_{n,i}=\mathcal{F}_i,
\end{gather*}
where ${\cal F}_i$ is the natural $\sigma$-field associated to the
model until time $i$. Then, we have
$E[T_{n,i}\mid\mathcal{G}_{n,i-1}]=0$,
$\mathcal{G}_{n,i}\underline{\subset}\mathcal{G}_{n+1,i}$ and
$T_n=\sum_{i=1}^n T_{n,i}$.  Thus, by a martingale central limit
theorem (see \cite{HH}),
$T_n\stackrel{d}\longrightarrow\mathcal{N}\bigl(0,\,\lambda(\beta)\bigr)$ 
  provided
\begin{gather*}
\text{(i) }\sum_{i=1}^n T_{n,i}^2\overset{P}\longrightarrow\lambda(\beta),\quad
\text{(ii) }\max_{1\leq i\leq n}|T_{n,i}|\overset{P}\longrightarrow 0,\quad
\text{(iii) }\sup_n E\left[\max_{1\leq i\leq n} T_{n,i}^2\right]<\infty.
\end{gather*}

Let
\begin{gather*}
D_i=(N_i-\Lambda_{i-1})^2\quad\text{and}\quad 
U_n=\frac{\sum_{i=1}^n\bigl\{D_i-E[D_i\mid\mathcal{F}_{i-1}]\bigr\}}{a_n(\beta)}
=\frac{\sum_{i=1}^n (D_i-\Lambda_{i-1})}{a_n(\beta)}.
\end{gather*}

By the same martingale argument used in the proof of the previous
theorem and by Kronecker's lemma, $U_n\overset{a.s.}\longrightarrow
0$.  Then, by (\ref{lim0}),
\begin{gather*}
\sum_{i=1}^n T_{n,i}^2=\frac{\sum_{i=1}^n D_i}{a_n(\beta)}=
U_n+\frac{\sum_{i=1}^n\Lambda_{i-1}}{a_n(\beta)}
\overset{a.s.}\longrightarrow\lambda(\beta).
\end{gather*}

This proves condition (i). As to (ii), fix $k\geq 1$ and note that
\begin{gather*}
\max_{1\leq i\leq n} T_{n,i}^2\leq
\frac{\max_{1\leq i\leq k} D_i}{a_n(\beta)}+
\max_{k<i\leq n}\,\frac{D_i}{a_i(\beta)}\leq
\frac{\max_{1\leq i\leq k} D_i}{a_n(\beta)}+
\sup_{i>k}\,\frac{D_i}{a_i(\beta)}\quad\text{for }n>k.
\end{gather*}
Hence, $\limsup_n\max_{1\leq i\leq n} T_{n,i}^2\leq\limsup_n\frac{D_n}{a_n(\beta)}$ 
and condition (ii) follows since

\begin{gather*}
\frac{D_n}{a_n(\beta)}=
\frac{\sum_{i=1}^n D_i}{a_n(\beta)}-\frac{\sum_{i=1}^{n-1} D_i}{a_n(\beta)}
\overset{a.s.}\longrightarrow 0.
\end{gather*}
Finally, condition (iii) is a consequence of 
\begin{gather*}
E\left[\max_{1\leq i\leq n} T_{n,i}^2\right]\leq
\frac{\sum_{i=1}^n E[D_i]}{a_n(\beta)}=
\frac{\sum_{i=1}^n E[\Lambda_{i-1}]}{a_n(\beta)} = \\
=\frac{\Lambda_0+\sum_{i=1}^{n-1} E[\Lambda_{i}]}{a_n(\beta)}
\leq\frac{\alpha\left(1+\sum_{i=1}^{n-1} (vi)^{\beta-1}\right)}{a_n(\beta)}.
\end{gather*}
\end{proof}

\subsection{Analysis of 
the random fitness parameters $R_i$}

Now our purpose is to find, under the assumption of our model, a
procedure to get information on the random variables $R_i$ from the
data, that typically are the values of $Z_1,\dots, Z_n$, i.e., $n$
rows of the matrix $Z$, where $n$ is the number of the observed nodes.

Unfortunately, this goal is not easily tractable as we will point out
in the sequel. The method we empirically tested extracts from the
data, with a maximum log-likelihood procedure (see
Section~\ref{sec:monte-carlo-experiments}), a plausible realization
$\widehat r_1,\dots, \widehat r_{k_n}$ of $R_1,\dots,R_{k_n}$, for a
suitable $k_n$; this information could be useful, for instance, to
reconstruct the ranking induced by $R_i$.  Note that we ideally would
like to find a probable realization for all the fitness parameters of
the observed nodes (not only for the first $k_n$ nodes), but we do not
possess the same amount of information about all $R_i$: in particular,
while $R_1$ influences all the subsequent observed rows of the matrix
$Z$, $R_{n-1}$ has only influence over $Z_n$. So we cannot expect to
find good values for all the random variables.

With the above purpose in mind, we now give a general expression for
the conditional probability of observing $Z_1=z_1,\dots, Z_n=z_n$
given $R_1,\dots, R_{n-1}$. We refer to Section \ref{sec-model} for the
notation.

The first row $Z_1$ is simply identified by $L_1=N_1$ and so
\begin{equation*}
\begin{split}
P(Z_1=z_1)&=P(N_1=n_1=card\{k: z_{1,k}=1\})\\
&=Poi(\alpha)\{n_1\}=e^{-\alpha}\frac{\alpha^{n_1}}{n_1!}.
\end{split}
\end{equation*}

Then the second row is identified by the values $Z_{2,k}$ with
  $k=1,\dots, L_1=N_1$ and by $N_2$ and so 
\begin{equation*}
\begin{split}
&P(Z_2=z_2|Z_1, R_1)=\\
&P(Z_{2,k}=z_{2,k}\,\hbox{for } k=1,\dots,L_1,\, 
N_2=n_2=card\{k>L_1: z_{2,k}=1\} |Z_1, R_1)
=\\
&\prod_{k=1}^{L_1} P_1(k)^{z_{2,k}}(1-P_1(k))^{1-z_{2,k}}
\times
Poi(\Lambda_1)\{n_2\},
\end{split}
\end{equation*}
where $P_1(k)$ is defined in (\ref{eq:define_inclusion_prob}) and 
$\Lambda_1$ is defined in (\ref{eq:define_lambda}).

The general formula is 
\begin{equation*}
\begin{split}
&P(Z_{j+1}=z_{j+1}|Z_1,R_1,\dots,Z_j,R_j)=\\ 
%%%&P\left(Z_{j+1,k}=z_{j+1,k}\,\hbox{for } k=1,\dots,L_j,\,
%%%N_{j+1}=n_{j+1}=card\{k>L_j: z_{j+1,k}=1\}|Z_1, R_1,\dots,Z_j,R_j\right)
%%%=\\ 
&P\left(Z_{j+1,k}=z_{j+1,k}\,\hbox{for } k=1,\dots,L_j,\,\right.
\\
&\quad\quad \left.N_{j+1}=n_{j+1}=card\{k>L_j: z_{j+1,k}=1\}
|Z_1, R_1,\dots,Z_j,R_j\right)
=\\ 
&\prod_{k=1}^{L_j} P_j(k)^{z_{j+1,k}}(1-P_j(k))^{1-z_{j+1,k}} \times
  Poi(\Lambda_j)\{n_{j+1}\},
\end{split}
\end{equation*}
where $P_j(k)$ is defined in (\ref{eq:define_inclusion_prob}) and 
$\Lambda_j$ is defined in (\ref{eq:define_lambda}).

Hence, for $n$ nodes, we can write a formula for the conditional
probability of observing $Z_1=z_1,\dots, Z_n=z_n$ given $R_1,\dots,
R_{n-1}$:
\begin{equation} \label{eq:prod_likelihood}
\begin{split}
&P(Z_1=z_1,\dots,Z_n=z_n|R_1,\dots,R_{n-1})=\\
&P(Z_1=z_1)
\prod_{j=1}^{n-1} P(Z_{j+1}=z_{j+1}|Z_1,R_1,\dots,Z_j,R_j).
\end{split}
\end{equation}

\subsubsection{A Monte Carlo method}
\label{sec:introducing-monte-carlo}

The algorithm we applied is essentially a MCMC (Markov Chain
Monte Carlo) method \cite{MCMC-ML}, which uses the basic principle of
Gibbs sampling \cite{Gibbs-tutorial}: fix all components of a vector
except one and compare the different values of the likelihood obtained
for various values of the non-fixed component.

The method employs the aforementioned formula
(\ref{eq:prod_likelihood}) for the conditional pro\-ba\-bi\-li\-ty of
observing $Z_1=z_1,\dots, Z_n=z_n$ given the values of $R_1,\dots,
R_{n-1}$. Precisely, using the symbol $\overline{z}$ in order to
denote the matrix with rows $z_1,\dots,z_n$ 
and the symbol $\overline{r}$ in order to denote a vector of component
$r_1,\dots, r_n$, set
\begin{equation}\label{l}
P(Z=
\overline{z}|R=\overline{r})=
P(Z_{1}=z_{1},\dots,Z_{n}=z_{n}|R_{1}=r_{1},\dots,R_{n}=r_{n}).
\end{equation}
We want to find a vector $\widehat{\overline r}$ that is a 
point maximizing the likelihood function (\ref{l}) corresponding to the
observed $\overline{z}$.
\footnote{We point out that our algorithm can not be considered a
  proper statistical estimation procedure for the fitness
  parameters. In particular, although it resembles the Bayesian {\em
    Maximum a posteriori probability} (MAP) estimation when the a
  priori distribution is an (improper) uniform distribution, we do not
  have a vector of parameters with a fixed dimension: the number of
  parameters in our case increases with the number of observations.}

The basic algorithm is described in Alg.~\ref{alg:montecarlo}. It is
regulated by these parameters:
\begin{itemize}
\item $\overline{r}^{0}\in\mathbb{R}^{n}$ is the initial guess for
  $\widehat{\overline r}$;

\item $J\in\mathbb{N}^{+}$ is the number of ``jumps to a new value'',
  i.e., the number of the new values analyzed for a certain component
  at each step;

\item $\sigma\in\mathbb{R}^{+}$ is the standard deviation of each ``jump''.

\end{itemize}

\begin{algorithm}[H]
\begin{center}

\framebox{\parbox[t]{0.8\columnwidth}
{

\textsc{Input}: $z_{1},\dots,z_{n}$, the observed features of each of
the $n$ observed nodes, i.e., the first $n$ rows of the attribute matrix
$Z$

\textsc{Output}: $\widehat{\overline r}$, a maximum point for the
likelihood function associated to the input data

\textsc{Description:}
\begin{enumerate}
\item $\widehat{\overline r}\leftarrow\overline{r}^{0}$
\item Repeat the following loop until convergence:

\begin{enumerate}
\item Choose a random node $i\in\left\{ 1,\dots,n\right\} $

\item Extract $J$ values $h_{1},\dots,h_{J}$ from the normal distribution
  $\mathcal{N}(r_{i},\sigma^{2})$; re-sample each $h_{j}$
  until $h_{j}>0$.

\item For each value $h_{j}$, compute 
\[
\mathcal{L}(h_{j})=
P(Z=\overline{z}|R_{1}=r_{1},\dots,R_{i}=h_{j},\dots,R_{n}=r_{n})\]

\item $\widehat r_{i}\leftarrow
\underset{h\in\left\{ r_{i},h_{1},\dots,h_{J}\right\} }
{\arg\max}\mathcal{L}(h)$
\end{enumerate}
\end{enumerate}
}}

\caption{\label{alg:montecarlo} 
Basic Monte Carlo algorithm to find $\widehat{\overline r}$. }
\end{center}
\end{algorithm}

It is worth to note that, given $\mathcal{L}(r_{i})$, it is possible
to find $\mathcal{L}(h)$ without re-doing the whole computation. In
fact, let us consider the product in eq. (\ref{eq:prod_likelihood}): a
change from $r_{i}$ to $h$ must be taken into account only from the
$i$-th factor onward -- that is, for the factors that come after
$P(Z_{i}=z_{i}|Z_{1},R_{1},\dots,Z_{i-1},R_{i-1})$.  In particular,
let $\delta=h-r_{i}$; then, for each $j$-th factor, with $j\geq i$, we
have to:

\begin{itemize}
\item add $\delta$ to the term $\sum_{i=1}^{j} R_{i}$, inside
  $\Lambda_{j}$ and $P_{j}(k)$ (defined in
  eq. \eqref{eq:define_inclusion_prob} and \eqref{eq:define_lambda});
\item add $\delta$ to the numerator of $P_{j}(k)$ when $k$ is
  s.t.~$z_{i,k}=1$; that is to say, change the global weight of a
  feature only if the node we changed has that feature.
\end{itemize}

Every other term in the equation remains unchanged and does not need to
be computed again. This remark allows us to speed up the implementation
con\-si\-de\-ra\-bly.\\

Figure~\ref{fig:gibbs1-likelihood} confirms that the algorithm moves
toward a vector $\widehat{\overline r}$ maximizing
$P(Z=\overline{z}|R=\overline{r})$ and shows
that the algorithm effectively converges. As a stopping criterion, we
can use the maximum increase in the log-likelihood in the last
iterations: when this is under a certain threshold $t$, we stop the
algorithm. The obtained outputs will be discussed in details in
Section \ref{sec:monte-carlo-experiments}.

\begin{figure}
	\begin{center}
		\begin{tabular}{cc}
			\footnotesize $R_{i}$ sampled from a uniform &
			\footnotesize $R_{i}$ sampled from a uniform discrete 
			\tabularnewline
			\footnotesize distribution on the interval~$[0.5,1.5]$ &
			\footnotesize distribution on the two 
values~$\{0.25,1.75\}$
			\tabularnewline
		\includegraphics[width=0.45\columnwidth]{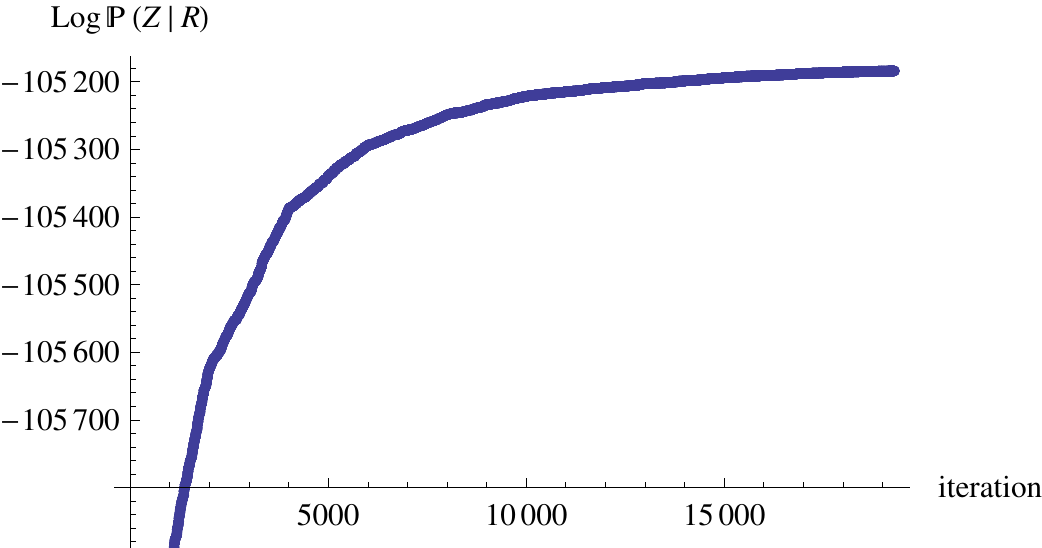}
	&  \includegraphics[width=0.45\columnwidth]{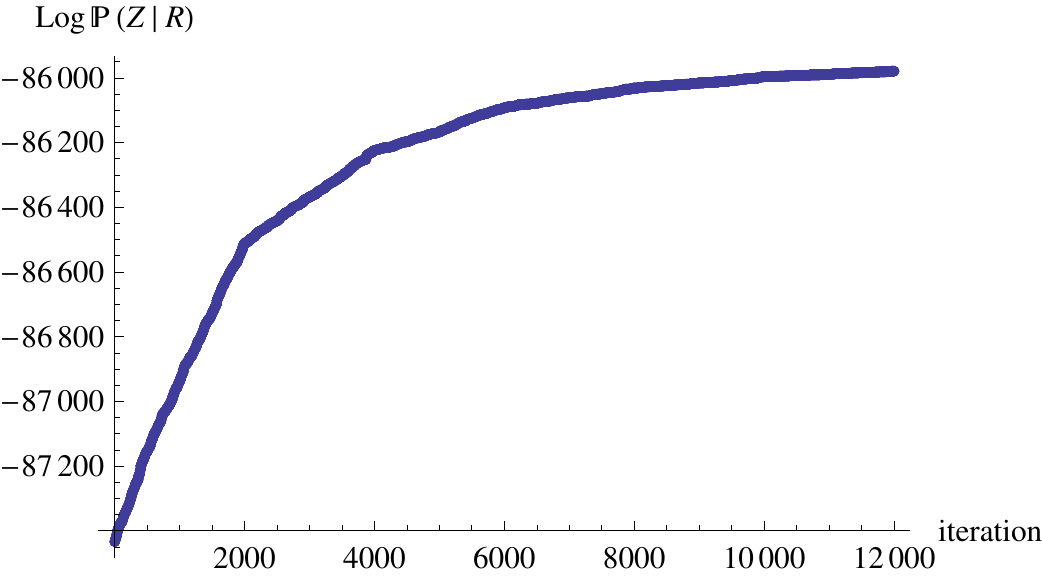}
			\tabularnewline
		\end{tabular}
		\caption{\label{fig:gibbs1-likelihood} Value of the
                  log-likelihood during the execution of the
                  algorithm, for different distributions of $R_{i}$.
                  The chosen algorithm parameters are $\sigma^{2}=1$,
                  $J=4$ and $\overline{r}^{0}=\mathbf{1}$ (the vector
                  with all $1$'s). The matrix $Z$ has $2000$ rows (nodes) and
                  it was generated with $\alpha=3$ and $\beta=0.9$.}
	\end{center}
\end{figure}

As already said, one point that we need to keep in mind is that we do
not possess the same amount of information about all the random
variables $R_i$: in particular, while $R_1$ influences all the
subsequent rows of the matrix $Z$, $R_{n-1}$ has only influence over
the last one. So we cannot expect the output values to be very
accurate for the last segment. For this reason, we also implemented a
variant of the algorithm that considers only the first $k_n$
nodes. Thus, we have another algorithm parameter $k_n$ so that the
choice of the jumping node at step 2(a) is restricted to $i \in
\lbrace 1, \dots, k_n\rbrace$ and, finally, the output will be the
corresponding segment of $\widehat{\overline r}$, i.e., $\widehat
r_1,\dots,\widehat r_{k_n}$. This variant converges faster and
moreover it allows to use larger values of the algorithm parameter
$J$.

Another relevant point is that the parameters $\alpha$
and $\beta$ enter the expression (\ref{eq:prod_likelihood}).
Therefore, in practice, before applying the algorithm, we need to
estimate them. As shown in Remark \ref{remark-estimator-beta}, we are
able to estimate $\beta$ starting from the observed values of the
matrix $Z$. On the other hand, as shown in Remark
\ref{remark-estimator-alpha}, the estimation of $\alpha$ presupposes
the knowledge of the mean value $m_R$ of the fitness parameters $R_i$
(except for the special case $\beta=1$).  Hence, we are in the
situation in which, in order to get information on the fitness
parameters by the proposed algorithm, we need to estimate $\alpha$ and
$\beta$, but, in order to estimate $\alpha$, we need to know the mean
value $m_R$. This problem can be partially solved as follows. Since
the term $P(Z_1=z_1)$ does not contain the $R_i$'s, the research of a
vector $\widehat{\overline r}$ that maximizes (\ref{l})
is equivalent to the research of a vector $\widehat{\overline r}$ maximizing
the product $$
\prod_{j=1}^{n-1} P(Z_{j+1}=z_{j+1}|Z_1,R_1,\dots,Z_j,R_j)
$$ 
in formula (\ref{eq:prod_likelihood}). On the other hand, each term of
the above product contains the inclusion probabilities $P_j(k)$, that
are invariant with respect to the normalization of the $R_i$'s by their
mean value $m_R$, and the $\Lambda_j$'s that have the property
$$
\Lambda_j=f(\alpha,\beta, {\overline r})=
f(\alpha/(m_R)^{1-\beta},\beta, {\overline r}/m_R)
$$ 
(where ${\overline r}/m_R$ denotes the vector with components
$r_i/m_R$). Consequently, starting from the observed values of the
matrix $Z$, we can 
\begin{itemize}
\item first, estimate $\beta$ by Remark \ref{remark-estimator-beta}; 
\item then estimate $\alpha'=\alpha/(m_R)^{1-\beta}$ by Remark
  \ref{remark-estimator-alpha} (i.e., $\widehat\alpha'_n$ equal to
  $\widehat\gamma_n$ or $\beta\,\widehat\gamma_n$ according to the
  estimated value of $\beta$);
\item finally, extract a plausible realization $\widehat{\overline
  r}'=\widehat{\overline r}/m_R$ (of the random variables
  $R_i'=R_i/m_R$) as a maximum point of the corresponding expression
  of the likelihood with the estimated value of $\beta$ and $\alpha'$. 

  Therefore, the output of the algorithm will be $\alpha'$, $\beta$
  and a plausible realization $\widehat{\overline r}'$ of the random
  variables $R_i'=R_i/m_R$.
\end{itemize}
\indent Finally, we highlight that it is possible to experiment other
variants of the algorithm, for example, by using a distribution
different from the normal for the jumps, or changing $\sigma$ during
the execution (e.g., reducing it according to some ``cooling
schedule'', as it happens in simulated
annealing~\cite{simulated-annealing}).  Additionally, instead of
looking for the values on the whole positive real line, we could
restrict the research on a suitable interval (guessed for the
particular real case).

\section{Simulations for the attribute matrix}
\label{sec:sim-Z}

In this section, we shall present a number of simulations we performed 
in order to illustrate the role of the parameters of the model and
also to see how good the proposed tools turn out to be.

\subsection{Estimating $\alpha$ and $\beta$}

Firstly, we aim at pointing out the role played by the model
parameters $\alpha$ and $\beta$. Therefore, we fix a distribution
for the random fitness parameters with $m_R=1$ and we simulate the
matrix $Z$ for different values of $\alpha$ and $\beta$ (fixing one
and making the other one change). More precisely, we assume that the
random variables $R_n$ are uniformly distributed on the interval
$[0.25, 1.75]$.

\begin{figure}
	\begin{center}
		\includegraphics[width=0.75\columnwidth]{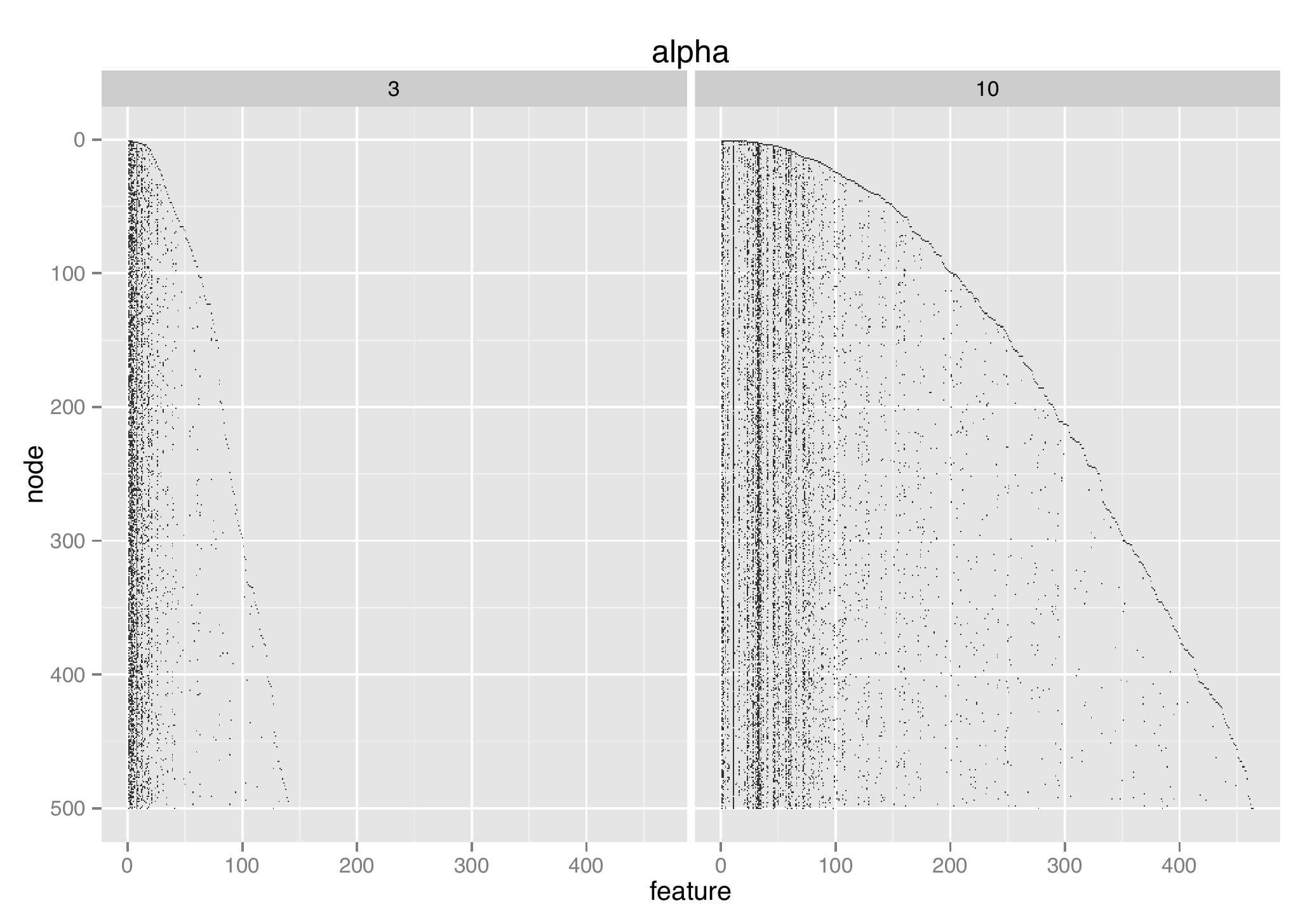}
		\caption{\label{fig:alpha-zmatrix} The $Z$ matrix for
                  $n=500$, two different values of $\alpha$
                  ($\alpha=3$ and $\alpha=10$) and a fixed
                  $\beta=0.5$. The random variables $R_i$ are
                  uniformly distributed on the interval $[0.25,
                    1.75]$.}
\end{center}
\end{figure}

In Figure~\ref{fig:alpha-zmatrix}, we visualize the effect of
$\alpha$: a larger $\alpha$ yields a larger number of new attributes
per node.

In Figure~\ref{fig:slope-beta}, instead, we visualize how different
positive values of $\beta$ yield a different power-law (asymptotic)
behavior of $L_n$. Indeed, in this figure, we have the log-log plot
of $L_n$ as a function of $n$. In the first two panels, we present two
different positive values of $\beta$ ($0.75$ and $0.5$), showing the
correspondence with the power-law exponent of $L_n$, estimated by the
slope of the regression line. Moreover, in the third panel, we point
out that the parameter $\alpha$ do not affect the power-law exponent
of $L_n$.

Figure~\ref{fig:estimator-beta} underlines that the estimator proposed
in remark \ref{remark-estimator-beta} works better (i.e. with a more
precision) for large values of $\beta$ since $L_n$ reaches the
power-law behavior more quickly for larger values of $\beta$.

\begin{figure}[htb]
	\begin{center}
	\includegraphics[width=1.0\columnwidth]{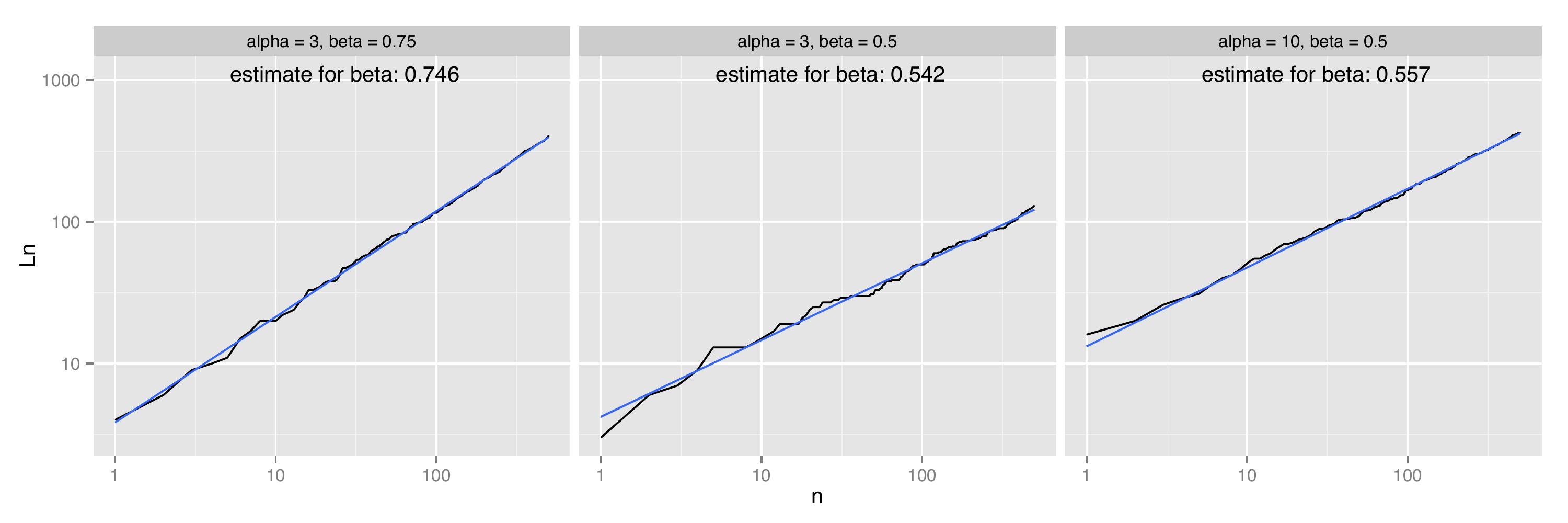}
	\caption{\label{fig:slope-beta} Correspondence between the
          parameter $\beta$ and the power-law exponent of $L_n$ as a
          function of $n$. The estimate of $\beta$ is the slope of the
          regression line. Here, we have $500$ nodes and the random
          variables $R_i$ are uniformly distributed on the interval
          $[0.25, 1.75]$. Values for $\alpha$ and $\beta$ are
          indicated above; we can see how different values for
          $\alpha$ do not affect the power-law behaviour. }
\end{center}
\end{figure}

\begin{figure}[htb]
	\begin{center}
		\begin{tabular}{ c c }
  			$\beta=0.5$ & $\beta=0.75$ \\
  	\includegraphics[width=0.45\columnwidth]{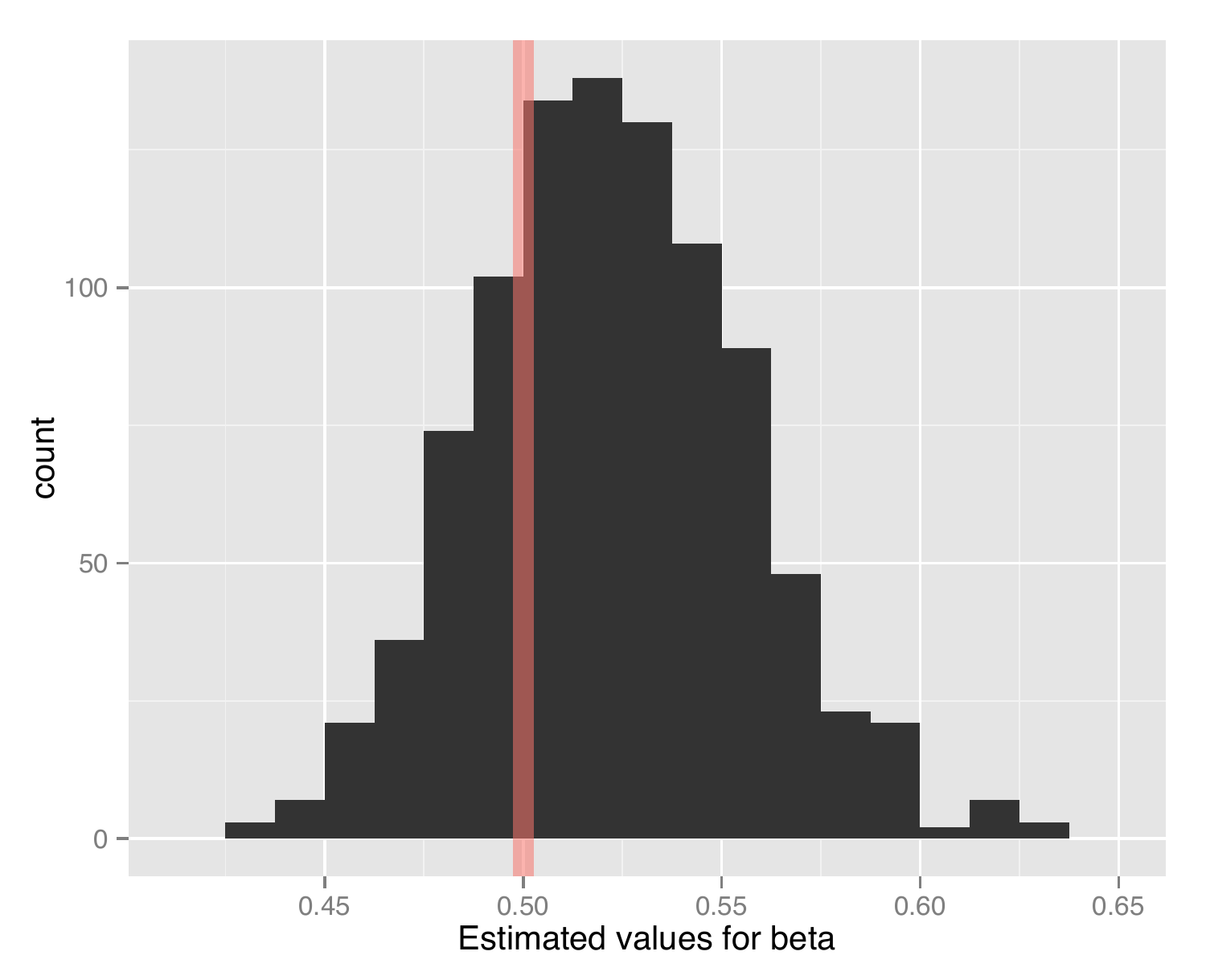} &  
\includegraphics[width=0.45\columnwidth]{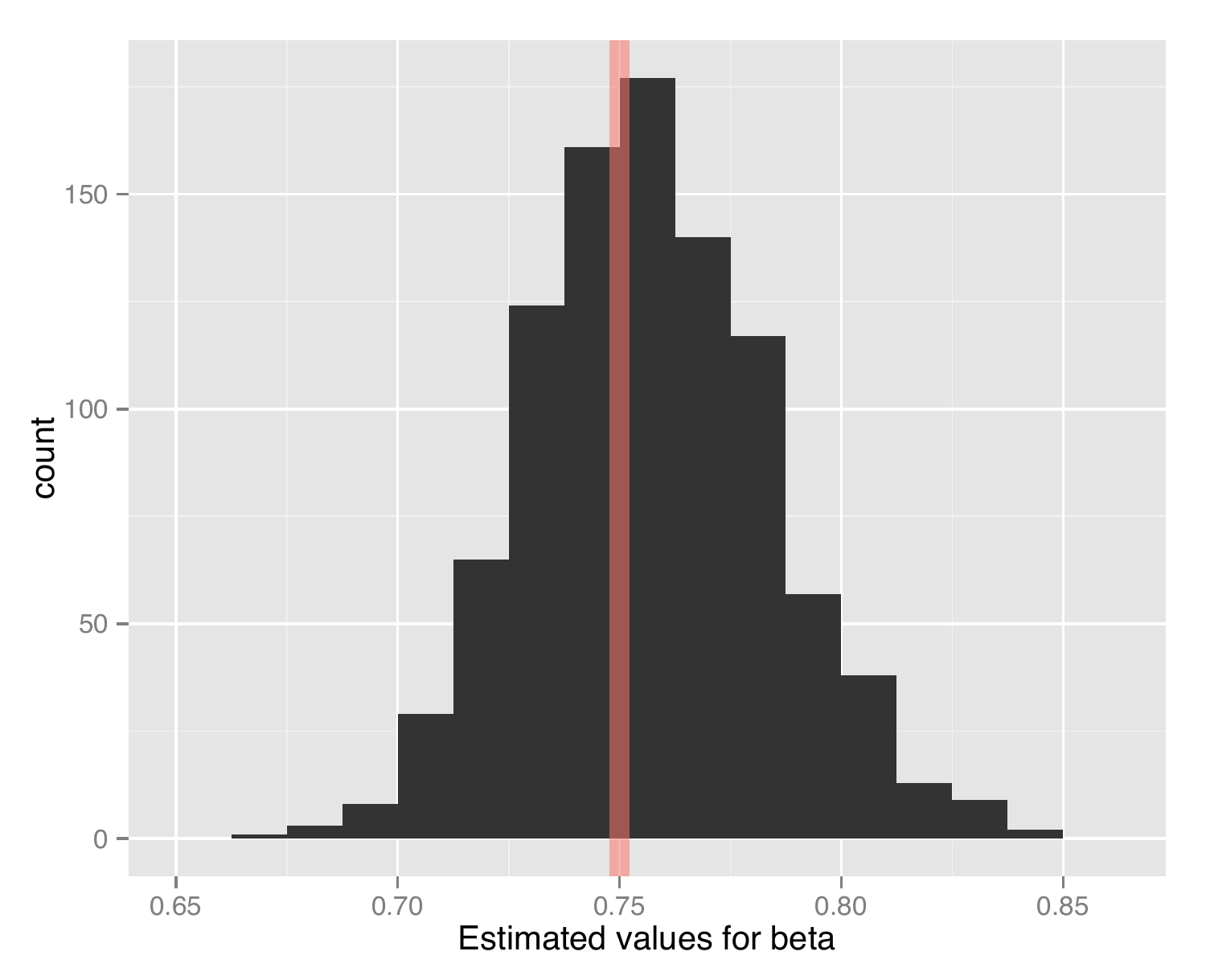} \\
		\end{tabular}
		\caption{ \label{fig:estimator-beta} Distribution of
                  the estimator $\widehat\beta_n$ of $\beta$ over
                  $1000$ experiments, each with $n=2000$ and
                  $\alpha=3$.  The random variables $R_i$ are
                  uniformly distributed on the interval $[0.25,
                    1.75]$.  The red line indicates the true value of
                  $\beta$. }
	\end{center}
\end{figure}

Similarly, we evaluated the estimator $\widehat\alpha_n$ of $\alpha$,
obtained by using the slope of the regression line in the plot of
$L_n$ as a function of $n^\beta$, as said in Remark
\ref{remark-estimator-alpha} (note that we have $m_R=1$ and so
$\alpha$ coincides with $\alpha'$).  Results are illustrated in Figure
\ref{fig:estimator-alpha} and show how this estimator yields good
results.

\begin{figure}[htb]
	\begin{center}
		\begin{tabular}{ c c }
  			$\alpha=3$ & $\alpha=10$ \\
  \includegraphics[width=0.45\columnwidth]{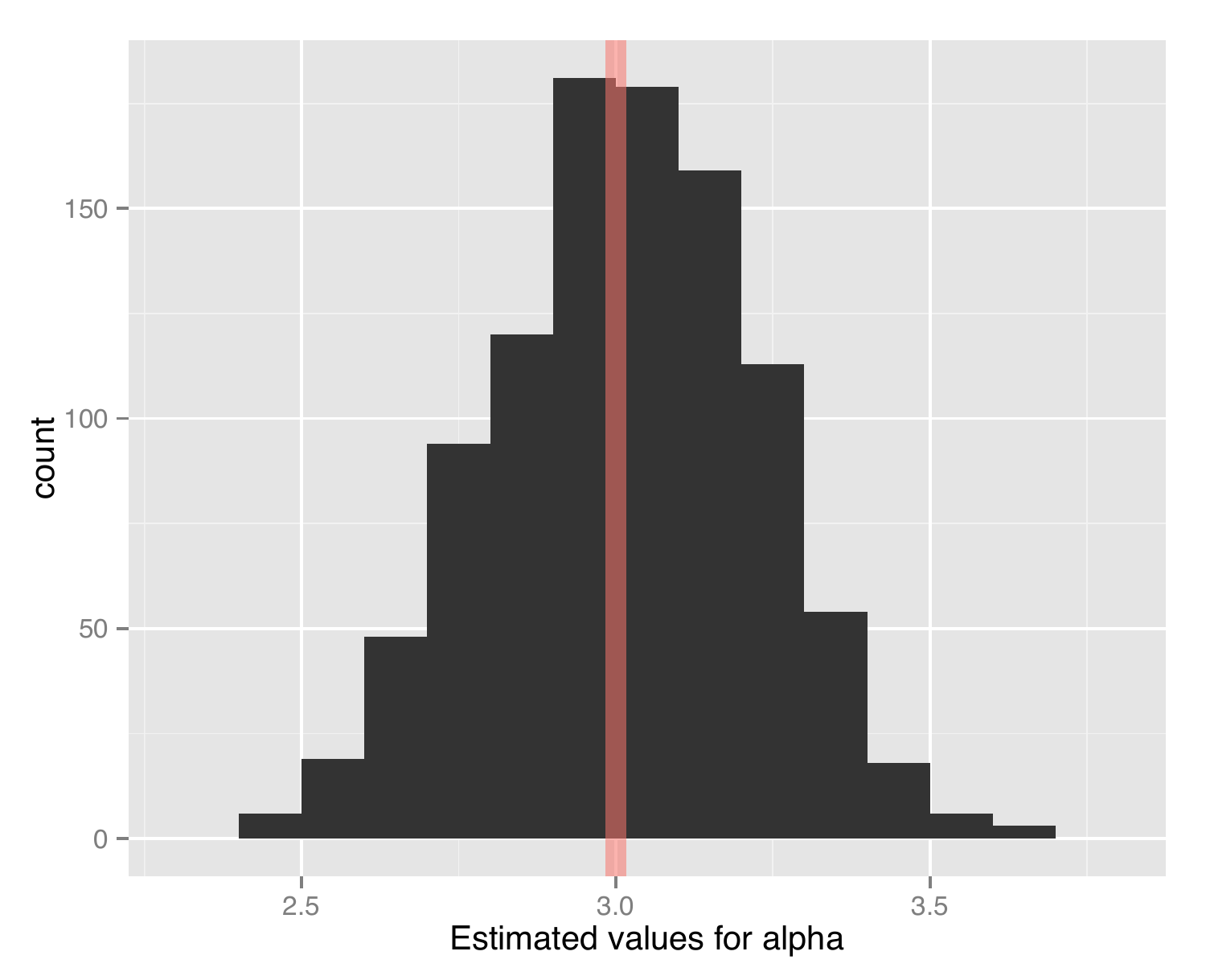} &  
\includegraphics[width=0.45\columnwidth]{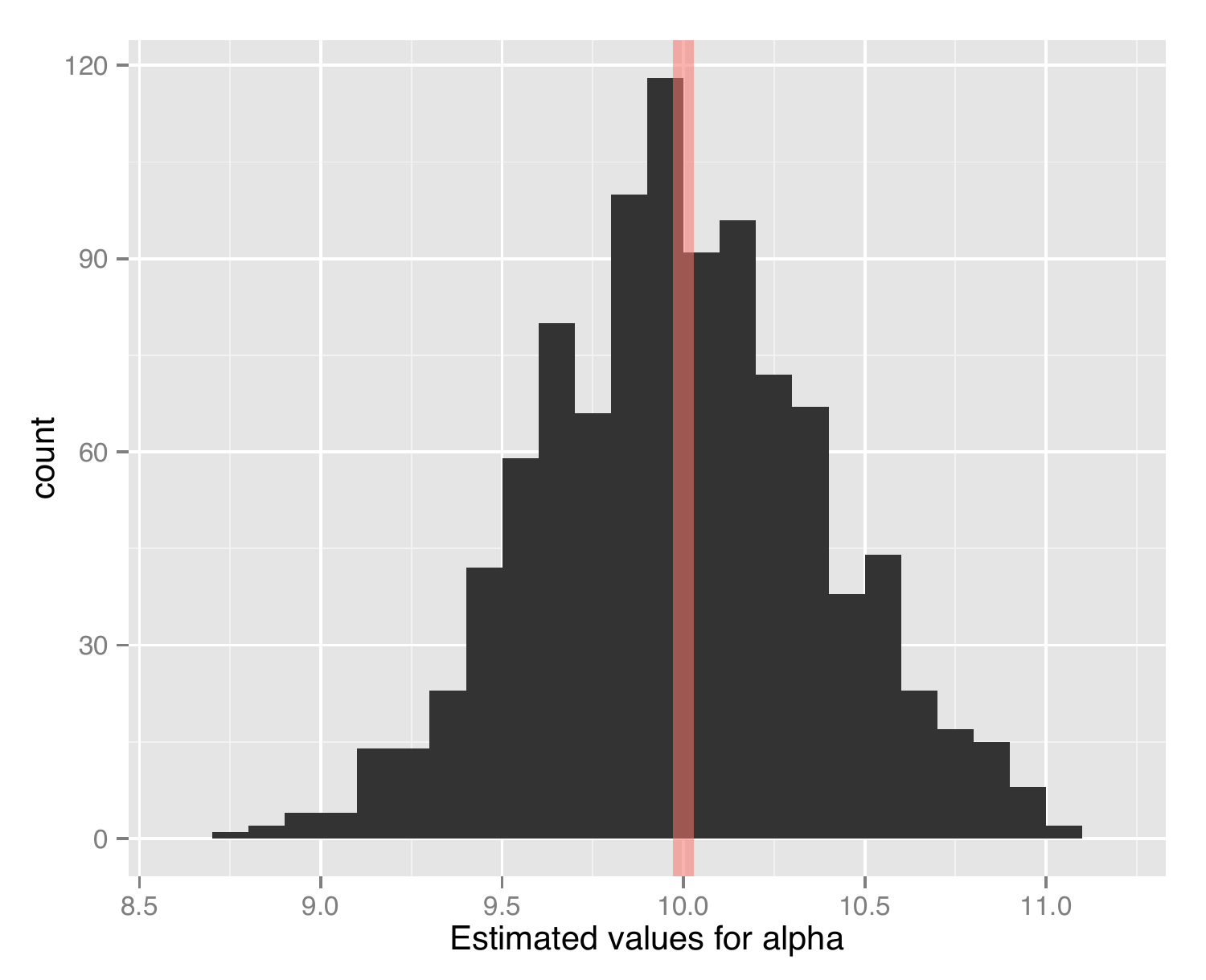} \\
		\end{tabular}
		\caption{ \label{fig:estimator-alpha} Distribution of
                  the estimator $\widehat\alpha_n$ of $\alpha$ over
                  $1000$ experiments, each with $n=2000$ and
                  $\beta=0.5$. The random variables $R_i$ are
                  uniformly distributed on the interval $[0.25, 1.75]$
                  (so $m_R=1$ and $\alpha=\alpha'$). The red line
                  indicates the true value of $\alpha$.}
\end{center}
\end{figure}

\begin{figure}[htb]
	\begin{center}
		\includegraphics[width=0.75\columnwidth]{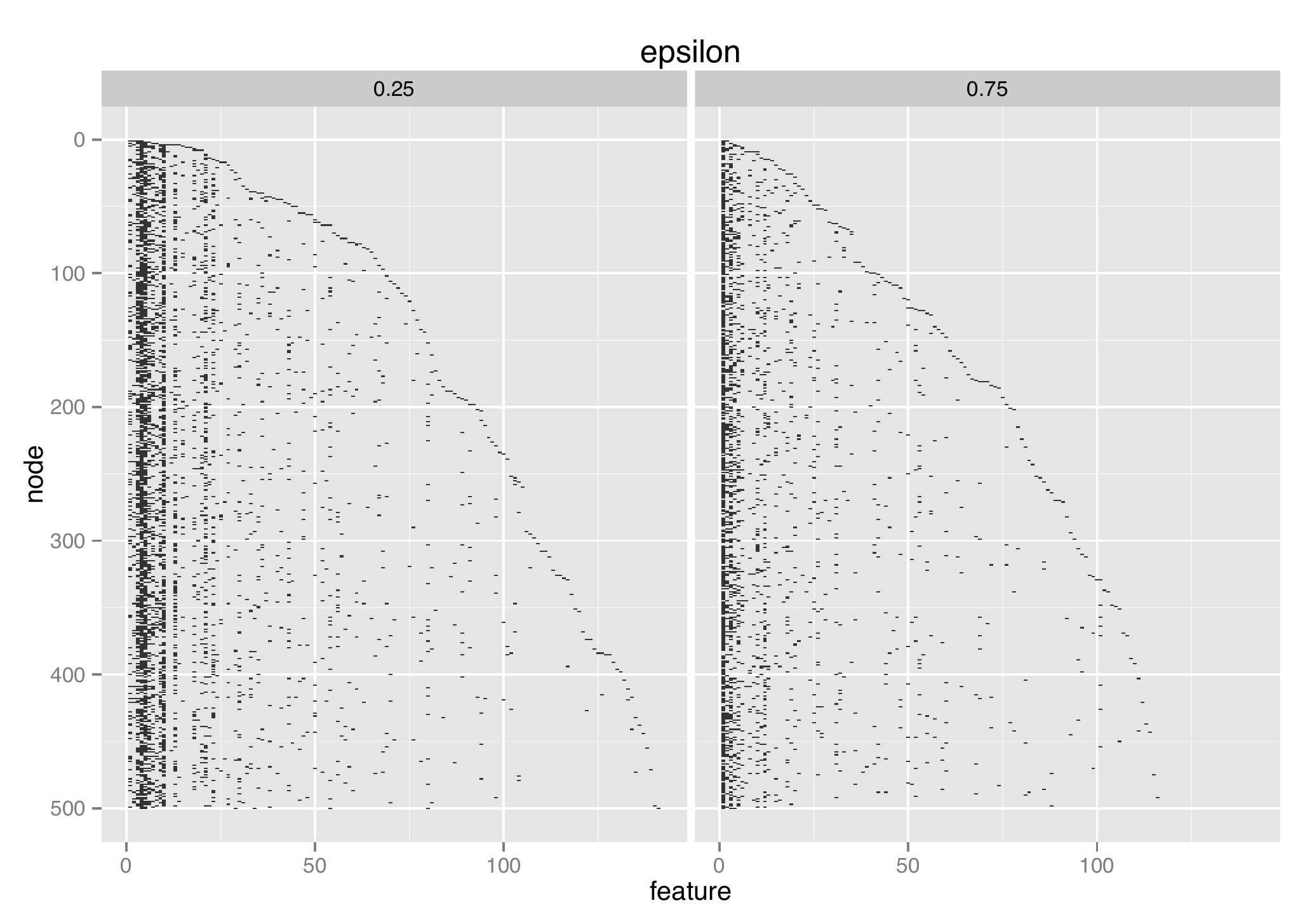}
		\caption{\label{fig:epsilon} Here $n=500$, $\alpha=3$,
                  $\beta=0.5$ and the random variables $R_i$ are
                  uniformly distributed on the interval
                  $[\epsilon,2-\epsilon]$ (so that $m_R=1$ and
                  $Var[R_i]=(1-\epsilon)^2/3$), for different values
                  of $\epsilon$ ($\epsilon=0.25$ and
                  $\epsilon=0.75$). The figure shows how $\epsilon$
                  affects the shape of $Z$.}
	\end{center}
\end{figure}

We also checked how the shape of the matrix $Z$ is influenced by the
distribution of the random parameters $R_n$. More precisely, we
analyzed the effect of $\epsilon$ on the shape of $Z$ when the random
variables $R_i$ are uniformly distributed on the interval
$[\epsilon,2-\epsilon]$, with $0<\epsilon\leq 1$, so that
$E[R_i]=m_R=1$ and the variance of $R_i$ is
$Var[R_i]=(1-\epsilon)^2/3$, which goes to zero as $\epsilon \to
1$. Hence, when $\epsilon$ is smaller, the variance of the $R_n$'s is
larger, so that also a ``young'' nodes $i$ have some chance of
transmitting their attributes to the other nodes (recall that a larger
$R_i$ makes $i$ more successful in transmitting its own
attributes). This is witnessed (see Figure~\ref{fig:epsilon}) by the
number of ``blackish'' vertical lines, that are more or less
widespread in the whole spectrum of nodes; whereas for larger
$\epsilon$ they are more concentrated on the left-hand side (i.e., only
the first nodes successfully transmit their attributes).

\subsection{Analysis of the random fitness parameters 
$R_i$}
\label{sec:monte-carlo-experiments}

We proceeded to test empirically how the Monte Carlo method performs
in re\-co\-ve\-ring the information on the fitness parameters
$R_{i}$. We tested its behavior against various distributions of
$R_{i}$; specifically, a uniform distribution on an interval, a
two-class uniform distribution, and finally a discrete power-law
distribution with $10$ classes. In the following of this section we
illustrate the details of such experiments, while, in the next
section, we will try to measure the performance of the proposed
technique. \\

\indent In every experiment, the matrix $Z$ has $n=2\,000$ nodes and it
was generated with $\alpha=3$ and $\beta=0.9$. The Monte Carlo
algorithm parameters were set as follows: $\sigma^{2}=1$, $J=4$ and
$\overline{r}^{0}=\mathbf{1}$ (the vector with all $1$'s). \\

For the first experiment, each $R_{i}$ is sampled from the uniform
distribution on the interval~$[0.5,1.5]$. We used the previously
discussed techniques to find the estimates of $\alpha$ and $\beta$:
the estimated values are $\widehat{\alpha}=3.095$ and
$\widehat{\beta}=0.893$ (note that we have $m_R=1$ and so
$\alpha=\alpha'$ and $\widehat{\overline r}=\widehat{\overline
  r}'$). Then, we tried the proposed Monte Carlo algorithm with the
stopping threshold $t={1}/{4}$. Results are visualized in
Figure~\ref{fig:gibbs1-results}, according to two different orderings
of the nodes:
\begin{itemize}
\item[i)] in the natural order, so that we confirm that our predictions
  are better for the first (i.e., the oldest) nodes than for the last
  (i.e., the youngest) ones;
\item[ii)] ordered by their true fitness values, so that we can show that
  we are, more or less, able to reconstruct the relative order of the
  fitness parameters (this fact will be made clearer in
  Section~\ref{sec:measure-tau}).
\end{itemize}

\begin{figure}[htb]
	\begin{center}
		\begin{tabular}{cc}
			in natural order & ordered by the value
			\tabularnewline
	\includegraphics[width=0.4\columnwidth]{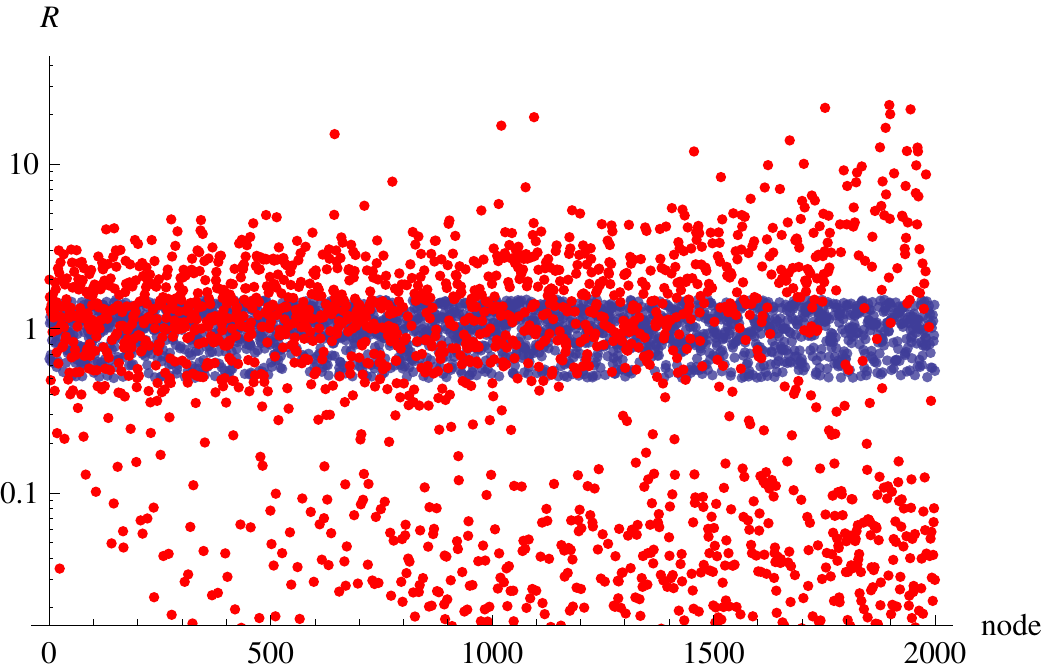} 
		& \includegraphics[width=0.4\columnwidth]{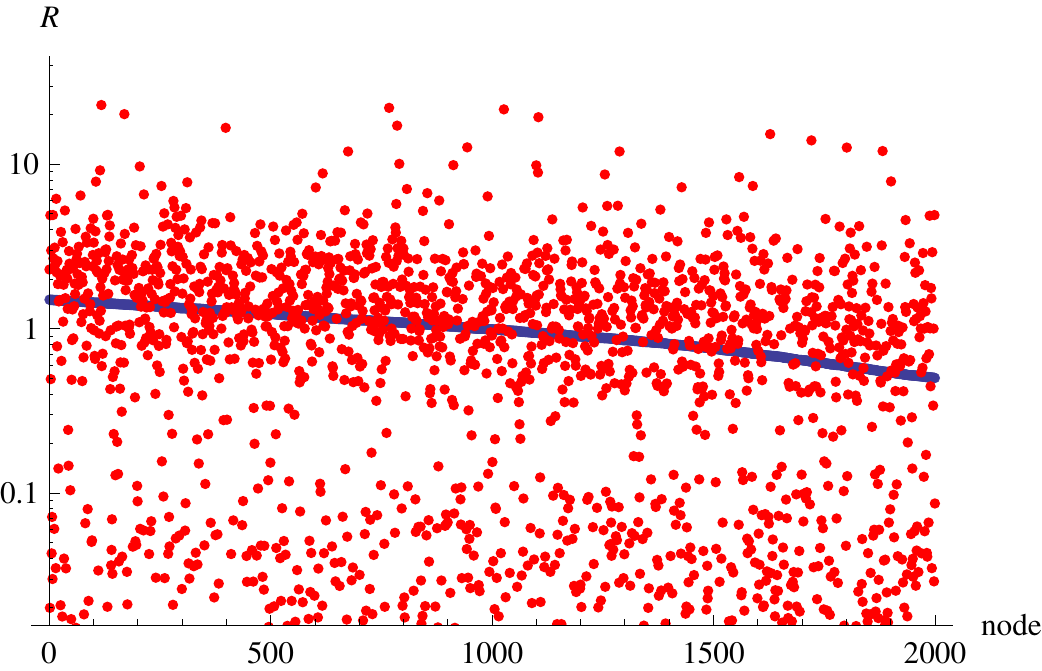}
			\tabularnewline
		\end{tabular}
		\caption{\label{fig:gibbs1-results} The extracted
                  realization $\widehat{\overline r}$ (in red) versus
                  the true realization $\overline{r}$ (in blue), with
                  two different orderings, in the case of uniform
                  distribution on the interval~$[0.5,1.5]$.
                  The empirical mean of the
         		   the first $\frac{n}{2}$ extracted values is $1.18$.}
	\end{center}
\end{figure}

In the second experiment, we applied our algorithm to a discrete case:
we sampled the fitness parameters $R_i$ from a set of only two values,
$\left\{0.25,1.75\right\}$, each with probability $\frac{1}{2}$.  We
left the parameters of the model and the ones of the algorithm
unaltered, except for moving the stopping threshold $t$ from $1/4$ to
$1$. The estimated values for $\alpha=3$ and $\beta=0.9$ are,
respectively, $\widehat{\alpha}=2.922$ (again $m_R=1$ and so
$\alpha=\alpha'$ and $\widehat{\overline r}=\widehat{\overline r}'$)
and $\widehat{\beta}=0.903$. The results of this second experiment are
more encouraging (we will see precise measurements in
Section~\ref{sec:measure-tau}). In this case, the output values of the
algorithm are closer to the true ones (see
Figure~\ref{fig:gibbs2-results}). Moreover, we can still observe the
same phenomena, i) and ii), described above.

\begin{figure}[htb]
	\begin{center}
		\begin{tabular}{cc}
			in natural order & ordered by value of $r_{i}$
			\tabularnewline
	\includegraphics[width=0.4\columnwidth]{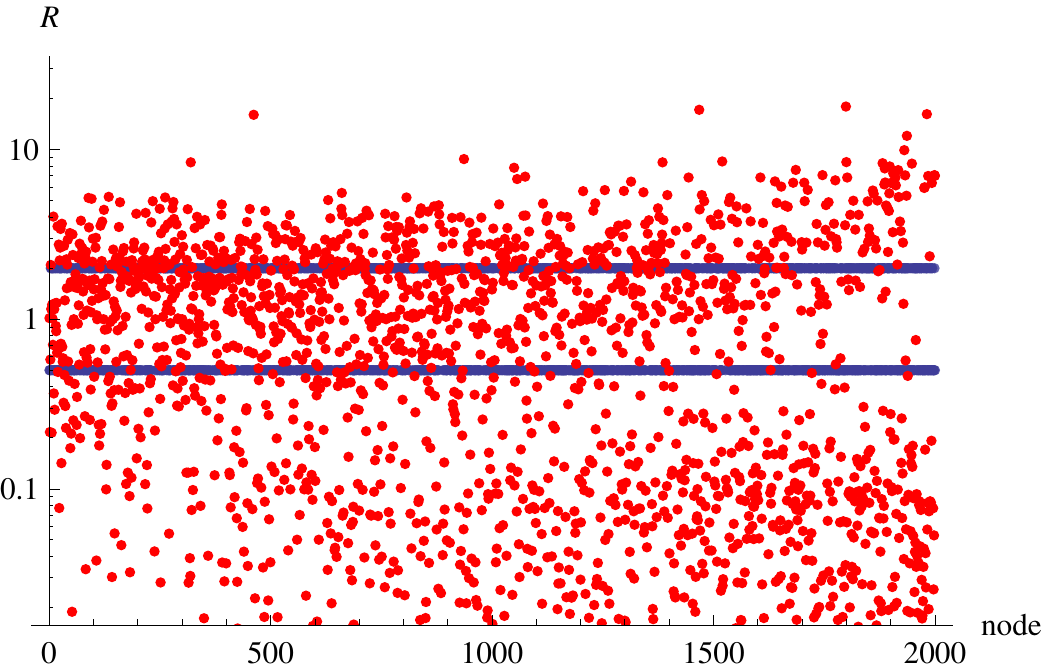} 
	& \includegraphics[width=0.4\columnwidth]{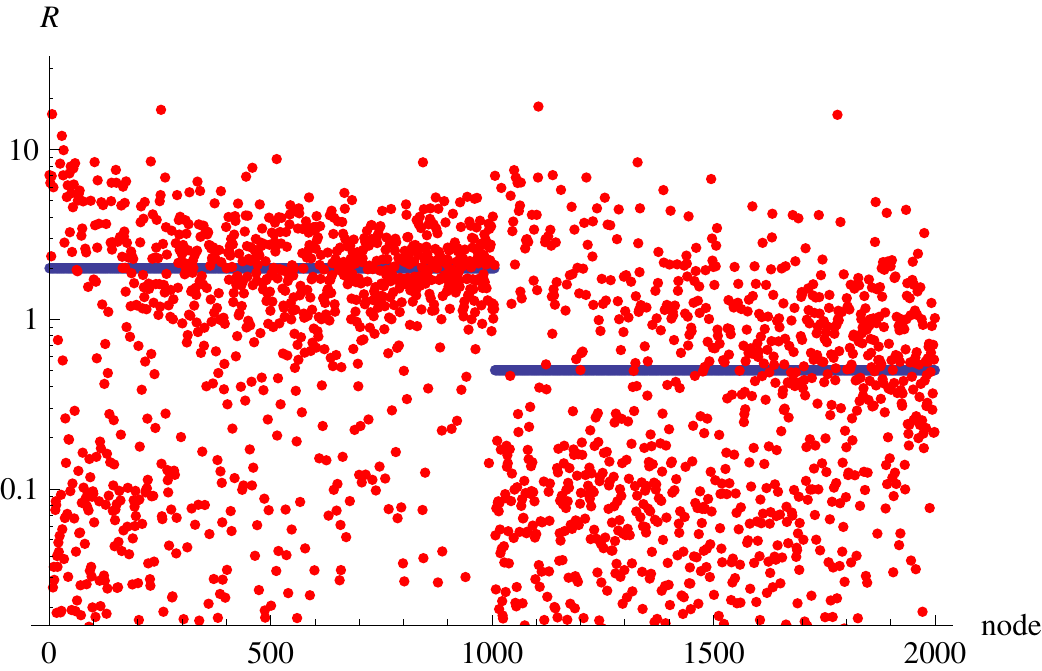}
			\tabularnewline
		\end{tabular}
	\caption{\label{fig:gibbs2-results} The extracted realization
          $\widehat{\overline r}$ (in red) versus the true realization
          $\overline{r}$ (in blue), with two different orderings, in
          the case of the uniform discrete distribution on the two
          values~$\{0.25,1.75\}$. The empirical mean of the
          the first $\frac{n}{2}$ extracted values is $1.33$.}
	\end{center}
\end{figure}

Finally, we applied the algorithm to a third case: we sampled $R_i$
from a normalized power-law discrete distribution, with $10$ possible
values -- specifically, a normalized discrete Zipf's law with exponent
$2$ and number of values $10$. We left both algorithm and model
parameters unaltered and we used $1$ as the stopping threshold $t$.

The estimated values for $\alpha=3$ and $\beta=0.9$ are, respectively,
$\widehat{\alpha}=3.595$ (again $m_R=1$ and so $\alpha=\alpha'$ and
$\widehat{\overline r}=\widehat{\overline r}'$) and
$\widehat{\beta}=0.868$.

Results for this case show that -- despite the fact that we have now a
discrete distribution with more than two values -- our approach can
recover information (especially for larger fitness values), as can
be seen in Figure~\ref{fig:gibbs-zipf-results} and in
Section~\ref{sec:measure-tau}.

\begin{figure}
	\begin{center}
		\begin{tabular}{cc}
			in natural order & ordered by the value
			\tabularnewline
\includegraphics[width=0.4\columnwidth]{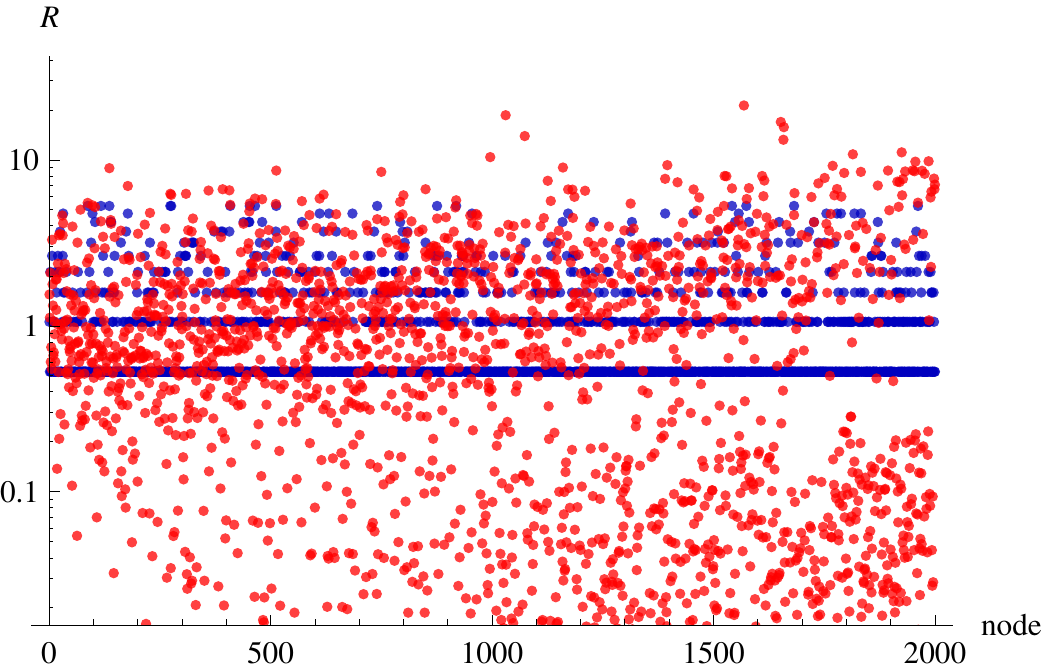} 
	& \includegraphics[width=0.4\columnwidth]{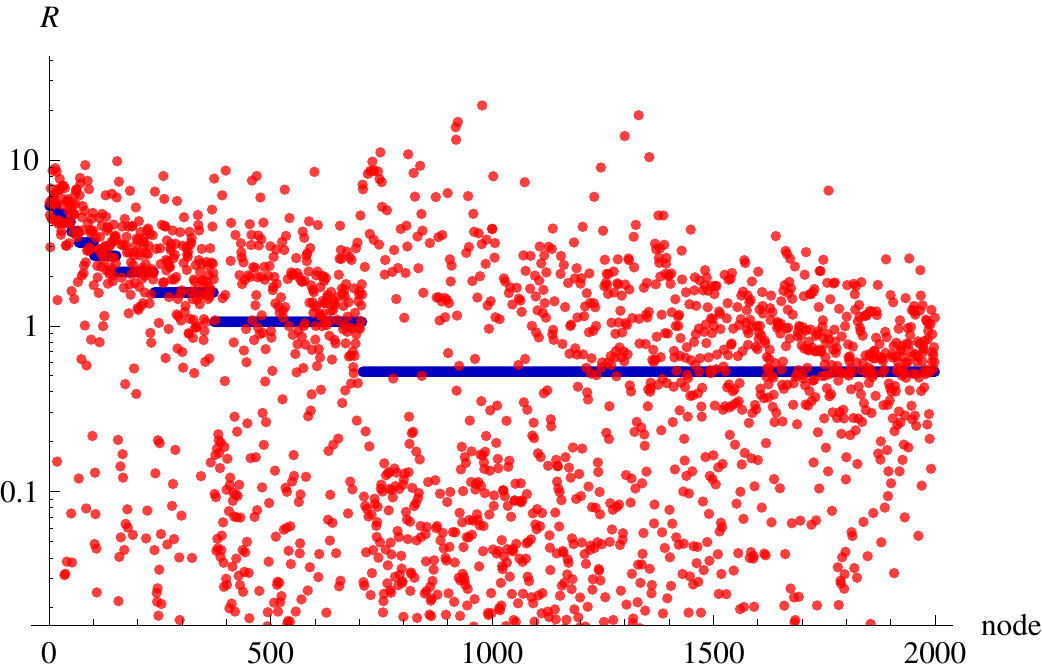}
			\tabularnewline
		\end{tabular}
		\caption{\label{fig:gibbs-zipf-results} The extracted
                  realization $\widehat{\overline r}$ (in red) versus
                  the true realization $\overline{r}$ (in blue), with
                  two different orderings, for the normalized discrete
                  Zipf's distribution with exponent $2$ and $10$
                  values. The empirical mean of the
         		  the first $\frac{n}{2}$ extracted values is $1.25$.}
	\end{center}
\end{figure}

We conclude this section noting that, for each of the experiments, the
Monte Carlo algorithm looks for the values of the fitness parameters
on the whole positive real line.  We would obtain better outputs if we
could restrict the research on a suitable interval for each case,
assuming a partial knowledge of the shape of the distribution.

\subsection{Analysis of the ordering of the nodes} 
\label{sec:measure-tau}

In a real application, we may content ourselves in finding not the
realized fitness parameters themselves but rather their ordering, that
is, the ordering of the nodes from larger to smaller values of the
fitness parameter.  To evaluate if we can at least extract values
$\widehat{r}_i$ that respect this ordering, we decided to compare the
drawn vector $\widehat{\overline{r}}$ with the true realization
$\overline{r}$ by the use of Kendall's $\tau$ and some variants of it.

To keep track of the fact that, as said before, the first nodes
contain more information than the last ones, we evaluated Kendall's
$\tau$ not only on the whole vector but also on a short initial
segment of size $k_n=n/2$ or $k_n=\sqrt{n}$. Besides this, we tried to
use a variant of Kendall's $\tau$ (proposed in \cite{vigna-tau}), that
we apply in two separate and different ways:

\begin{enumerate}
	\item inducing a hyperbolic decay based on the position of the
          nodes -- that is, weighting more the first (the oldest)
          nodes, and less the last (the youngest) ones;
	\item inducing a hyperbolic decay based on the true realized
          values $r_i$ -- that is, assigning a higher weight to the
          nodes with a greater fitness parameter $r_i$.
\end{enumerate}

\begin{table}
	\begin{center}
\begin{tabular}{|c|c|c|c|}
\hline 
Considered nodes & Kendall's $\tau$ & $\tau$ weighted by
position & $\tau$ weighted by value \tabularnewline \hline \hline
$k_n=\lfloor\sqrt{n}\rfloor=44$ & .281 & .206 & .463 \tabularnewline
\hline $k_n=\frac{n}{2}=1000$ & .229 & .188 & .337 \tabularnewline \hline
$k_n=n=2000$ & .150 & .139 & .155 \tabularnewline \hline
\end{tabular}
		\caption{\label{table:gibbs1-kendall} Comparing
                  orderings induced by the true realization
                  $\overline{r}$ versus the extracted one 
                  $\widehat{\overline{r}}$ in the case of the uniform
                  distribution on the interval~$[0.5,1.5]$.  }
	\end{center}
\end{table}

\begin{table}
	\begin{center}
\begin{tabular}{|c|c|c|c|}
\hline Considered nodes & Kendall's $\tau$ & $\tau$ weighted by
position & $\tau$ weighted by value \tabularnewline \hline \hline
$k_n=\lfloor\sqrt{n}\rfloor=44$ & .676 & .593 & .713 \tabularnewline
\hline $k_n=\frac{n}{2}=1000$ & .586 & .585 & .625 \tabularnewline \hline
$k_n=n=2000$ & .438 & .477 & .434 \tabularnewline \hline
\end{tabular}
		\caption{\label{table:gibbs2-kendall} Comparing
                  orderings induced by the true realization
                  $\overline{r}$ versus the extracted one 
                  $\widehat{\overline{r}}$ in the case of the uniform
                  discrete distribution on the two
                  values~$\{0.25,1.75\}$.  }
	\end{center}
\end{table}

\begin{table}
	\begin{center}
\begin{tabular}{|c|c|c|c|}
\hline Considered nodes & Kendall's $\tau$ & $\tau$ weighted by
position & $\tau$ weighted by value \tabularnewline \hline \hline
$k_n=\lfloor\sqrt{n}\rfloor=44$ & .735 & .762 & .772 \tabularnewline
\hline $k_n=\frac{n}{2}=1000$ & .453 & .516 & .803 \tabularnewline \hline
$k_n=n=2000$ & .313 & .402 & .543 \tabularnewline \hline
\end{tabular}
		\caption{\label{table:gibbs-zipf-kendall} Comparing
                  orderings induced by the true realization
                  $\overline{r}$ versus the extracted one
                  $\widehat{\overline{r}}$ in the case of the
                  normalized discrete Zipf's distribution with
                  exponent $2$ and $10$ values.  }
	\end{center}
\end{table}

The results of these measures are summarized in Table
\ref{table:gibbs1-kendall} for the experiment with the uniform
distribution on an interval, in Table \ref{table:gibbs2-kendall} for
the experiment with the uniform discrete distribution on the two
values~$\{0.25,1.75\}$, and in Table \ref{table:gibbs-zipf-kendall} for
the discrete Zipf's distribution with $10$ values and exponent $2$.
The tables show that, although we are unable to reconstruct the actual
realized values of the fitness parameters, our approach actually
recovers some information about node ranking.  As already seen before,
the output of the Monte Carlo algorithm is better for the discrete
cases.

\section{From the attribute structure to the graph}
\label{sec:graph}

We now extend the model to produce a graph out of the attribute
structure (that may itself be latent and unknown). In general, we may
assume that the presence of an edge between two nodes depends on the
features that those nodes exhibit, but there are many nuances to this
idea and possible approaches.

In the sequel, we postulate that the connections are undirected (we
omit self-loops, i.e., edges of type $(i,i)$) and we denote the
adjacency matrix (symmetric by assumption) by $A$. 
%%%
%%%we assume that,
%%%conditioned on $Z$ (and some other variables), the probability of
%%%having at time $n$ a certain adjacency matrix (symmetric by
%%%assumption) $a=(a_{i,j})_{1\leq i,j\leq n}$ (with $a_{i,j}\in\{0,1\}$)
%%%is
%%%
%%%\begin{equation*}
%%%\begin{split}
%%%P(A=a|Z,\, \hbox{other variables})&=
%%%P\big( \bigcap_{1\leq j<i\leq n}\{A_{i,j}=a_{i,j}\}|Z,\,\hbox{other variables})
%%%\\
%%%&=\prod_{1\leq j< i\leq n} P(A_{i,j}=a_{i,j}|Z,\, \hbox{other variables}).
%%%\end{split}
%%%\end{equation*}

\subsection{Feature/Feature probability model (FF)}
In the first, basic model, we assume that the probability of having an
edge $(i,j)$ depends \emph{solely} on the features that $i$ and $j$
possess; each pair of feature vectors that node $i$ and node $j$ exhibit
contributes in tuning the edge probability.  In other words, letting
$L_n$ be the total number of different features (i.e., columns of
$Z$), we assume that there is a symmetric feature-feature influence
matrix $\Xi=(\xi_{h,k})_{1\leq h,k\leq L_n}$ that determines a
node-node weight matrix $W$ given by
\[
	W=Z \cdot \Xi \cdot Z^T
\]
or, more explicitly,
\[
	w_{i,j}=\sum_{h,k} Z_{i,h} \xi_{h,k} Z_{j,k}.
\] 
The probability of the presence of an edge $(i,j)$, then, depends
monotonically on $w_{i,j}$.  The choice of $\Xi$ determines different
relations between features and edge pro\-ba\-bi\-li\-ties.  If $\xi_{h,k}>0$
(resp., $\xi_{h,k}<0$), then the simultaneous presence of attributes
$h$ and $k$ increases (resp., decreases) the edge probability; if
$\xi_{h,k}=0$, the simultaneous presence of attributes $h$ and $k$
does not affect the edge-probability. In particular, if $\xi_{h,k}=0$
for $h\neq k$, then the edge-probability is affected only by the
presence of the same attributes in both nodes (positively or negatively
affected depending on the sign of $\xi_{h,h}$).

The actual probabilities are computed as some function applied to the
corresponding weight; i.e., some monotone function $\Phi: {\mathbf R} \to
[0,1]$ is fixed and, for $1\leq j<i\leq n$, 
\begin{equation}\label{FF-model}
	P(A_{i,j}=1|Z)=\Phi(w_{i,j})=\Phi\left(\sum_{h,k} Z_{i,h} \xi_{h,k}
	Z_{j,k}\right).
\end{equation}

\subsection{Feature/Feature+BA probability model (FFBA)}
A variant of the feature/feature (FF) probability model takes into
account the fact that some edges exist independently of the features
that the involved nodes exhibit, but they are there simply because of
the popularity of a node, as in the traditional ``preferential
attachment'' model by Barab\'asi and Albert~\cite{BarabasiAlbert}.  To
take this into consideration, instead of using (\ref{FF-model}), we
rather define for $1\leq j<i\leq n$
\begin{equation}\label{FF+BA-model}
	P(A_{i,j}=1|Z,\,D_j(i-1),\, m(i-1))=
\delta\Phi\left(\sum_{h,k} Z_{i,h} \xi_{h,k} Z_{j,k}\right)+
(1-\delta)\frac{D_j(i-1)}{2 m(i-1)},
\end{equation}
where $D_j(k)$ and $m(k)$ are, respectively, the degree of node $j$
and the overall number of edges just after node $k$ was added. The
parameter $\delta$ controls the mixture between the pure
feature/feature model and the preferential-attachment model
(degenerating to the first one when $\delta=1$, and to the second one
when $\delta=0$).

\subsection{Feature/feature+JR probability model (FFJR)}
Jackson and Rogers~\cite{JR} observed that preferential attachment can
be induced also injecting a ``friend-of-friend'' approach in the
creation of edges.  Their behavior can be mimicked in our model as
follows: we first generate a graph with adjacency matrix $A'$ using
the pure FF model, i.e., letting
\[
	P(A'_{i,j}=1|Z)=\Phi(w_{i,j})=\Phi\left(\sum_{h,k} Z_{i,h} \xi_{h,k}
	Z_{j,k}\right).
\]
After this, every node $i$ looks at the set of the neighbors of its
neighbors, according to $A'$. If this set is not empty, it then selects
one node from the set uniformly at random; the resulting node is chosen as an
``extra'' friend of $i$ with some probability $1-\delta$ (for suitably
chosen $\delta \in [0,1]$). The adjacency matrix obtained in this way
is $A$.  Once more, if $\delta=0$ we have $A=A'$ so we get back to the
pure FF model.

\section{Simulations for the graph structure}
\label{sec:graph-simulations}
The purpose of this collection of experiments is to determine the
topology of the graph generated with the models described above.  We
fix \emph{a priori} the number of nodes $n$ and the (approximate)
number of edges $m$ (i.e., density) we aim at; then, every experiment
consists essentially in two phases:
\begin{itemize}
  \item generating an attribute matrix $Z$ for $n$ nodes (with certain
    values for the parameters $\alpha$ and $\beta$ and with $R_i$
    uniformly distributed on the interval $[\epsilon,2-\epsilon]$
    for a certain $\epsilon$);
  \item building the graph according to one of the models described in
    Section~\ref{sec:graph}.
\end{itemize} 
The second phase needs to fix some further parameters: $\Xi$ (the
feature/feature influence matrix), the function $\Phi$ and, for the
mixed models (FFBA and FFJR), the parameter $\delta$.

For the sake of simplicity, throughout this section, we assume that
$\Xi=I$ and we take $\Phi$ as a sigmoid function given by
\[
	\Phi(x) = \frac{1}{e^{K(\theta-x)}+1}.
\]
In other words, the existence of an edge $(i,j)$ depends simply on the number
of features that $i$ and $j$ share (this is an effect of choosing
$\Xi=I$). More features induce larger probability: the sigmoid
function smoothly increases (from $0$ to $1$) around a threshold $\theta$,
and $K>0$ controls its smoothness; when $K\to \infty$ we obtain a step
function and edges are chosen deterministically based on whether the
two involved nodes share more than $\theta$ features or not.
 
In the experiments, we fix $K$ and determine $\theta$ on the basis of
the desired density of the graph (or, equivalently, the desired number
of edges $m$); in practice\footnote{The described method needs some
  (obvious) adjustments when applied to the mixed models, to take into
  account the edges generated by preferential attachment.}, this is
obtained by solving numerically the equation
\[
E\left[\sum_{1\leq j<i\leq n} A_{i,j}\right]
=	
\sum_{1\leq j<i\leq n}\Phi\left(\sum_{h,k} Z_{i,h} \xi_{h,k} Z_{j,k}\right) 
= 
m
\]
for the indeterminate $\theta$ (using, for example, Newton's
method). Since $\Xi=I$ the equation in fact simplifies into
\[
\sum_{1\leq j<i\leq n}\frac{1}{e^{K(\theta-\sum_{h} Z_{i,h} Z_{j,h})}+1} = m.
\]

\smallskip
With these assumptions, every experiment depends on the parameters
used for generating $Z$ (i.e., $\alpha$, $\beta$ and $\epsilon$), on
$K$ (that controls the smoothness of the sigmoid function) and on
$\delta$ (for the mixed models). In the graphs produced by each
simulation, we took into consideration the degree distribution, the
percentage of reachable pairs (i.e., the fraction of pairs of nodes
that are reachable) and the distribution of distances (lengths of
shortest paths); the latter data are computed using a probabilistic
algorithm~\cite{HyperANF}.

Some of the results obtained (for $n=2\,000$ and\footnote{We observed
absolutely analogous phenomena also for larger and denser networks;
we hereby report only the smaller case for the sake of readability
of the pictures.} $m=4\,000$) for the FF model are shown in
Figure~\ref{fig:expgraph-ff}. For those experiments, the underlying
attribute matrix is generated with $\beta=0.75$ and
$R_i$ uniformly distributed on the interval $[0.75, 1.25]$;
we compare $\alpha=3$ (resulting in $\approx
1200$ features) with $\alpha=10$ ($\approx 4000$ features).
Results regarding mixed models are reported in Figure~\ref{fig:expgraph-ffmix}.
  
The properties of the obtained graphs can be summarized as follows:
\begin{itemize}
  \item the pure FF model exhibits a behavior that strongly depends on the
  smoothness parameter $K$ (see Fig. \ref{fig:expgraph-ff}):
  \begin{itemize}
    \item for $K=1$, the degree distribution is power-law only when
      $\alpha$ is large (e.g., $\alpha=10$), whereas the distribution
      is often non-monotonic for smaller $\alpha$'s, especially on
      large graphs; the fraction of reachable pairs is quite large
      (between $40\%$ and $90\%$);
    \item for $K=4$, degrees are always distributed as a power-law
      (with exponents around $3$), but the graph becomes largely
      disconnected (the reachable pairs are never more than $20\%$):
      this is because nodes with the same degree tend to stick
      together (assortativity), forming a highly connected component
      and leaving the remaining nodes isolated;
    \item for $K\to \infty$, the power-law distribution of degrees is
      even more clear-cut, but the number of reachable pairs becomes
      smaller (no more than $10\%$); the exponent of the power-law
      distribution depends on $\alpha$, with larger $\alpha$'s
      yielding larger absolute values of the (negative) exponents;
  \end{itemize}
  \item the FFBA model (see Fig. \ref{fig:expgraph-ffmix}) increases
    slightly the number of reachable pairs in all cases; the shape of
    the power-law distribution is essentially unchanged with respect
    to the pure FF model;
  \item finally, for the FFJR model (see
    Fig. \ref{fig:expgraph-ffmix}) we observe a reduced connectivity;
    this is due to holding the expected number of edges as a constant,
    while devoting some of them to closing triangles -- an operation
    that cannot increase connectivity. The degree distribution seems
    closer to a power-law with respect to the pure FF model.
 \end{itemize}
 
 \newcolumntype{V}{>{\centering\arraybackslash} m{.75\linewidth} }
\newcolumntype{v}{>{\centering\arraybackslash} m{.1\linewidth} }

\begin{figure}[htb]
\begin{center}
\begin{tabular}{v|cV}
\multirow{3}{*}{} & $K=1$ & \includegraphics[width=0.75\columnwidth]
{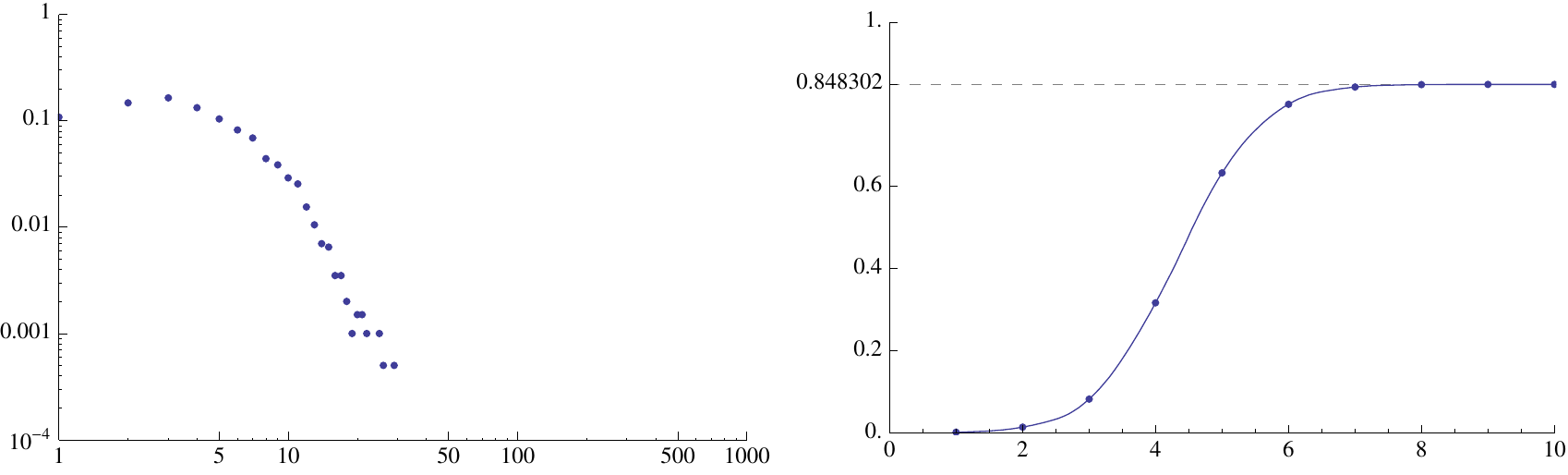} \\ 
$\alpha=3$ & $K=4$ & \includegraphics[width=0.75\columnwidth]
{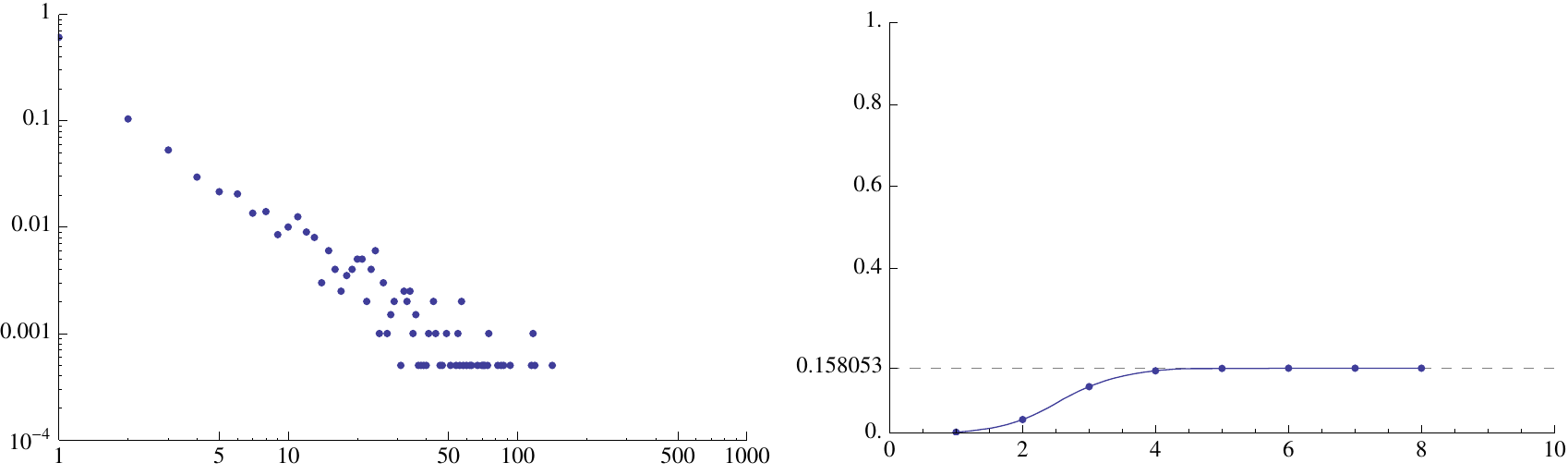} \\ 
& $K\rightarrow \infty$ & \includegraphics[width=0.75\columnwidth]
{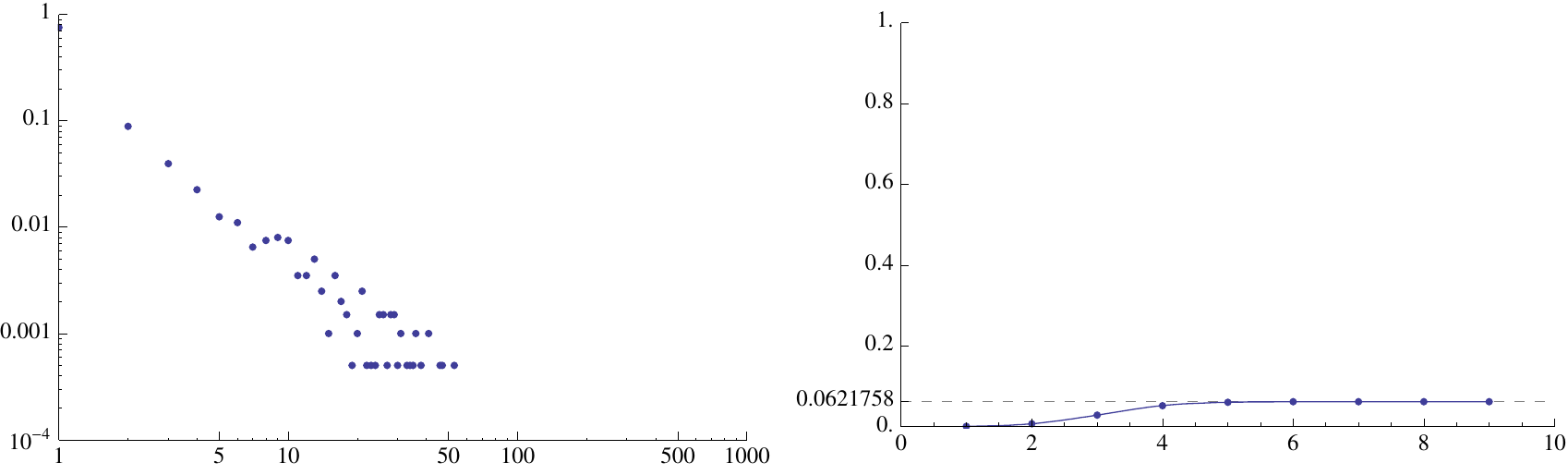} \\ 
\hline
\multirow{3}{*}{} & $K=1$ & \includegraphics[width=0.75\columnwidth]
{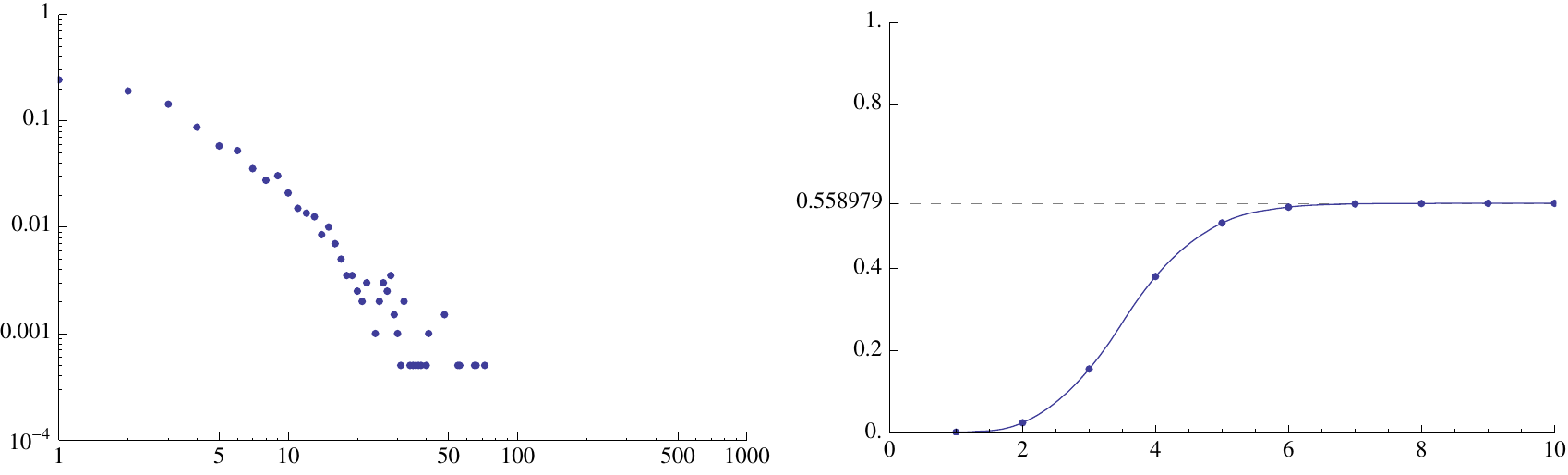} \\ 
$\alpha=10$ & $K=4$ & \includegraphics[width=0.75\columnwidth]
{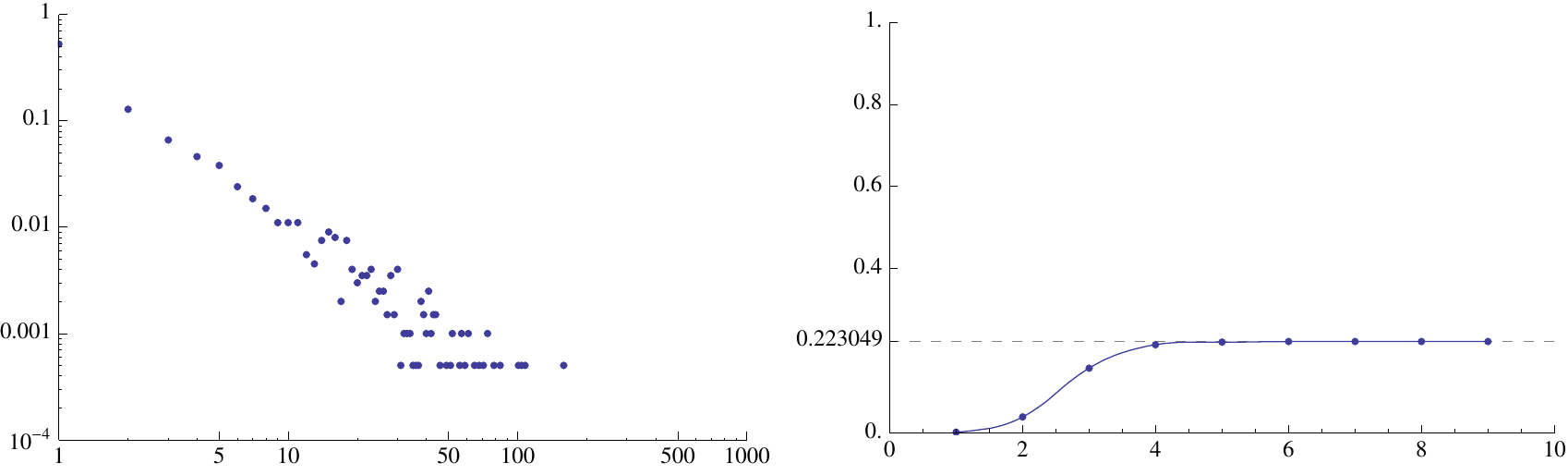} \\ 
& $K\rightarrow \infty$ & \includegraphics[width=0.75\columnwidth]
{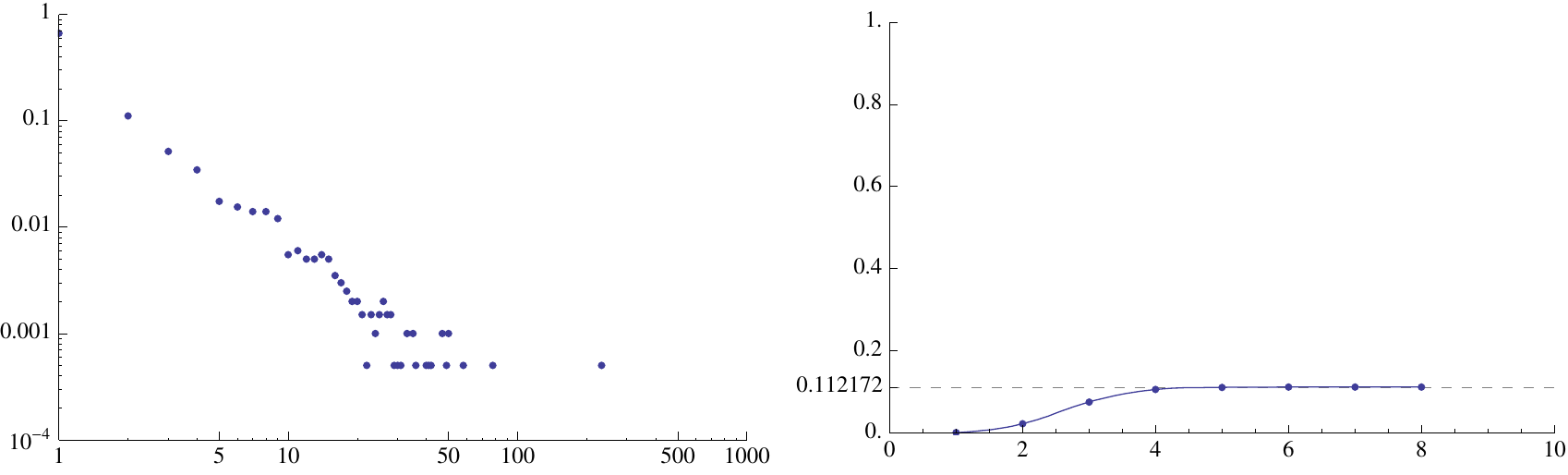} \\
\end{tabular} 
\caption{\label{fig:expgraph-ff} \footnotesize{Properties of graphs
    generated by the FF model.  We show the degree distribution in a
    log-log plot and the fraction of pairs at distance at most $k$; in
    the latter, we highlight the peak value (fraction of mutually
    reachable pairs).}}
\end{center}
\end{figure}

\begin{figure}[htb]
\begin{center}
\begin{tabular}{cV}
FFBA &
\includegraphics[width=0.75\columnwidth]{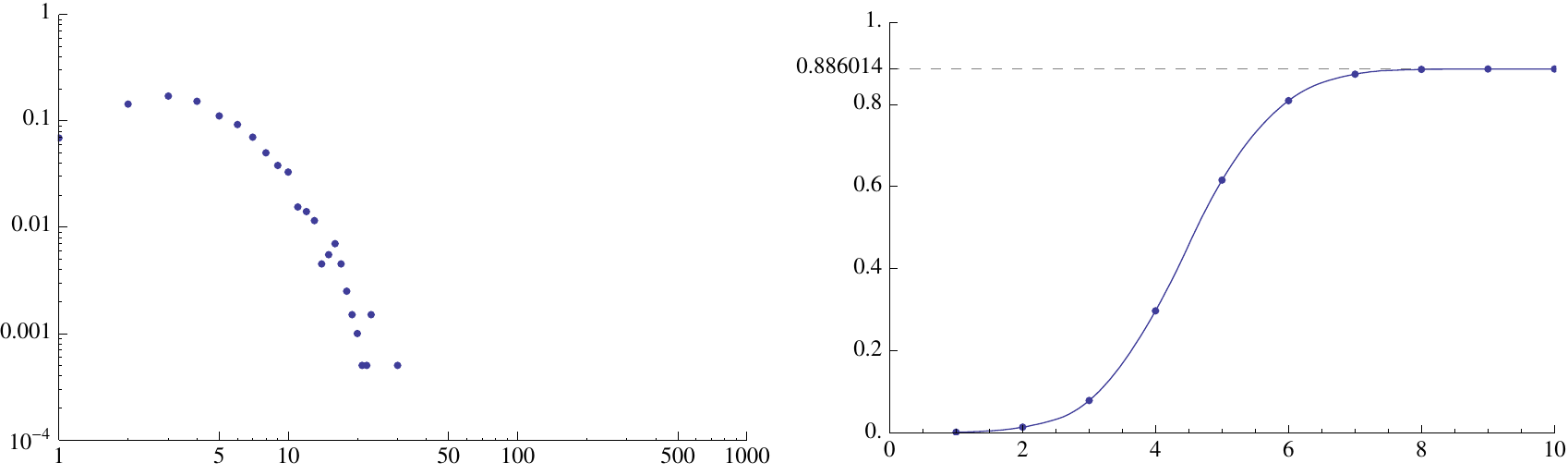}
\\
\hline
FFJR &
\includegraphics[width=0.75\columnwidth]
{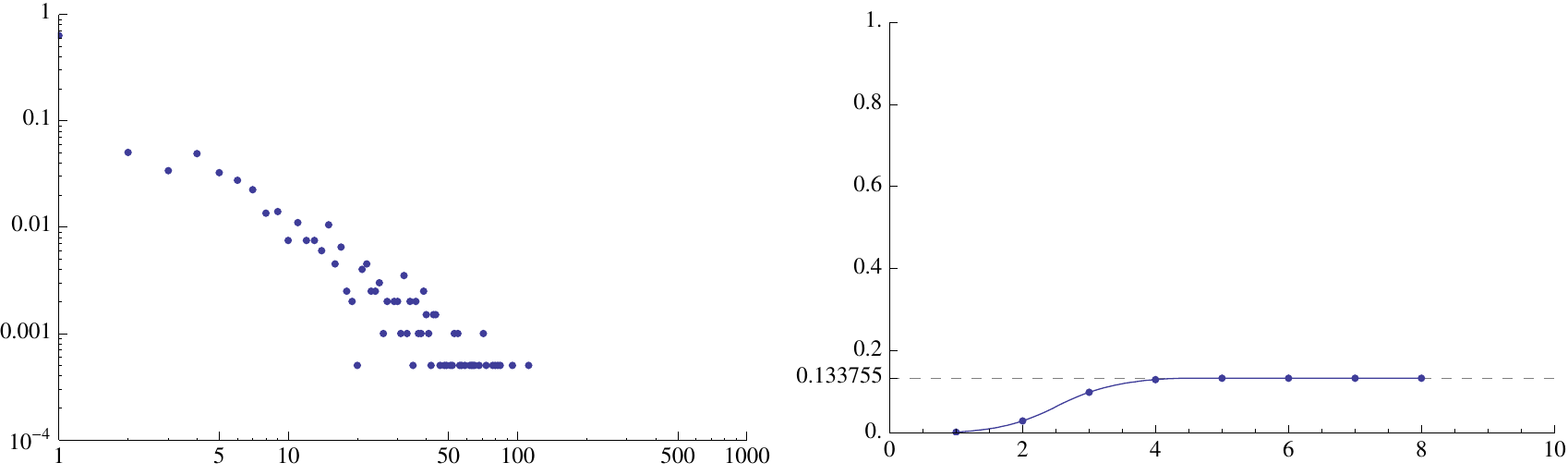}
\\
\end{tabular} 
\caption{\label{fig:expgraph-ffmix} Properties of graphs generated by
  mixed models with $K=1$ and $\delta=0.75$.  We show the degree
  distribution in a log-log plot and the fraction of pairs at distance
  at most $k$; in the latter, we highlight the peak value, indicating
  how many pairs of nodes are mutually reachable. The parameters of
  the underlying attribute-matrix model are $\alpha=3$ and
  $\beta=0.75$ and the $R_i$'s are uniformly distributed on
  $[0.75,1.25]$.}
\end{center}
\end{figure}

\section{A real dataset}
\label{real-data}

We considered a dataset of scientific papers\footnote{The dataset is
  available within the SNAP (Stanford Large Network Dataset
  Collection) at
  \texttt{http://snap.stanford.edu/data/cit-HepTh.html}.}  (originally
released as part of the 2003 KDD Cup) consisting of $27\,770$ papers
from the ``High energy physics (theory) arXiv'' database.  For each
paper (node), we considered as features the words appearing in its
title and abstract, excluding those that are dictionary
words\footnote{According to the Unix \texttt{words} dictionary.}.  The
papers are organized in order of publication date.

In Figure~\ref{fig:real-attribute-matrix} the reader can see a
fragment of the attribute matrix (for the first $500$ nodes and the
features they exhibit).

The overall number of features is $21\,933$, with a matrix density of
$0.35\cdot 10^{-3}$ (there are $214\,510$ ones in the matrix). The
estimated values of $\alpha'$ and $\beta$ are $15.038$ and $0.671$,
respectively.  In particular, we recall that $\beta$ is the power-law
exponent of the asymptotic behavior of $L_n$, i.e. the overall number
of distinct attributes. We show the estimate for this real case in
Figure~\ref{fig:slope-beta-real-case}. A recostruction of the ordering
of the nodes according to their fitness parameter values is possible,
but we lack any ground truth to compare it to.

\begin{figure}[H]
\centering
\begin{subfigure}{.45\textwidth}
  \centering
  \includegraphics[scale=0.5]{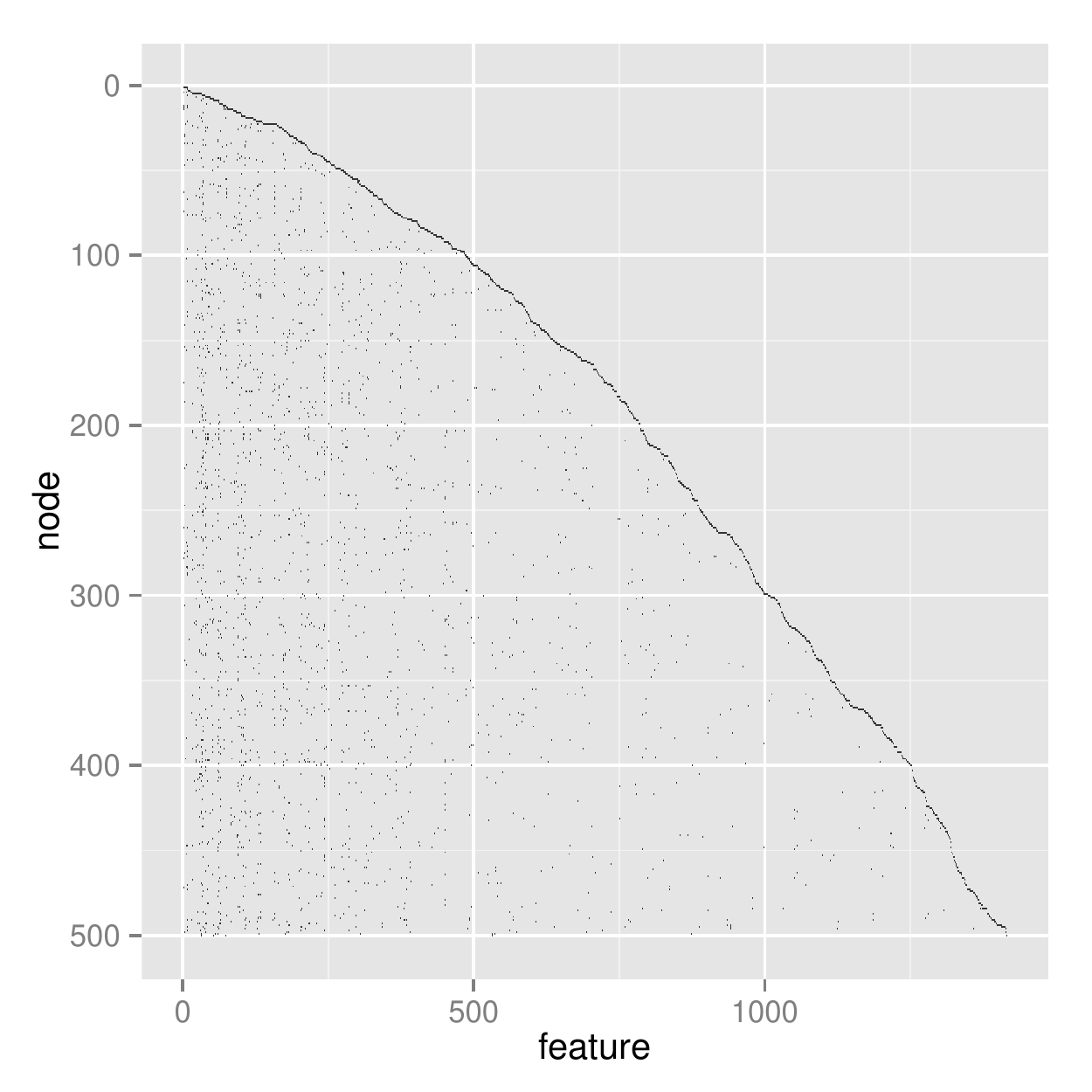} 
  \caption{\footnotesize  The first $500$ rows (nodes) of the attribute matrix.}
  \label{fig:real-attribute-matrix}
\end{subfigure}
\quad
\begin{subfigure}{.45\textwidth}
  \centering
  \includegraphics[scale=0.4375]{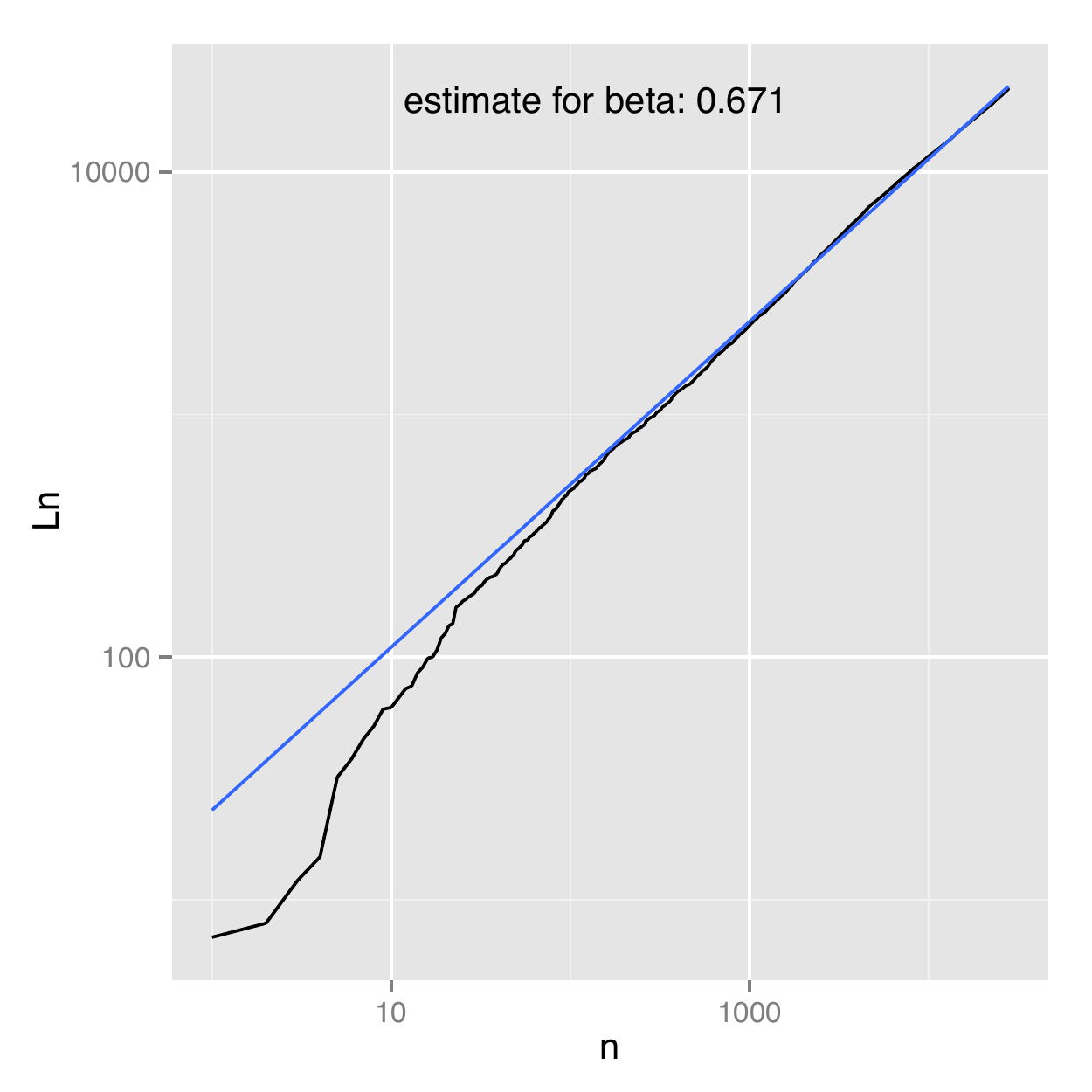}
  \caption{\footnotesize Correspondence between the parameter $\beta$
    and the power-law exponent of $L_n$, as a function of $n$. The
    estimate of $\beta$ is the slope of the regression line.}
  \label{fig:slope-beta-real-case}
\end{subfigure}
\caption{Analysis of the \texttt{cit-HepTh} dataset. }
\end{figure}

We conclude this section with a comparison between the graph produced
by the FF model using as the underlying matrix the attribute matrix of
the \texttt{cit-HepTh} dataset and the corresponding (symmetrized)
citation graph. After some experiments, we observed that we can obtain
a good fit with $K=2.5$, that produces a quite similar degree and
distance distribution (see Figure~\ref{fig:graph-analysis-real}). It
is striking to observe that the two graphs have such a strong
similarity in their topology, albeit having positively no direct
relation with each other (in one case the edges represent citations,
in the other they were obtained by the model basing on the textual
similarity of their abstracts!).

\begin{figure}[htb]
\begin{center}
\begin{tabular}{cV}
\texttt{cit-HepTh} 
& \includegraphics[width=0.75\columnwidth]{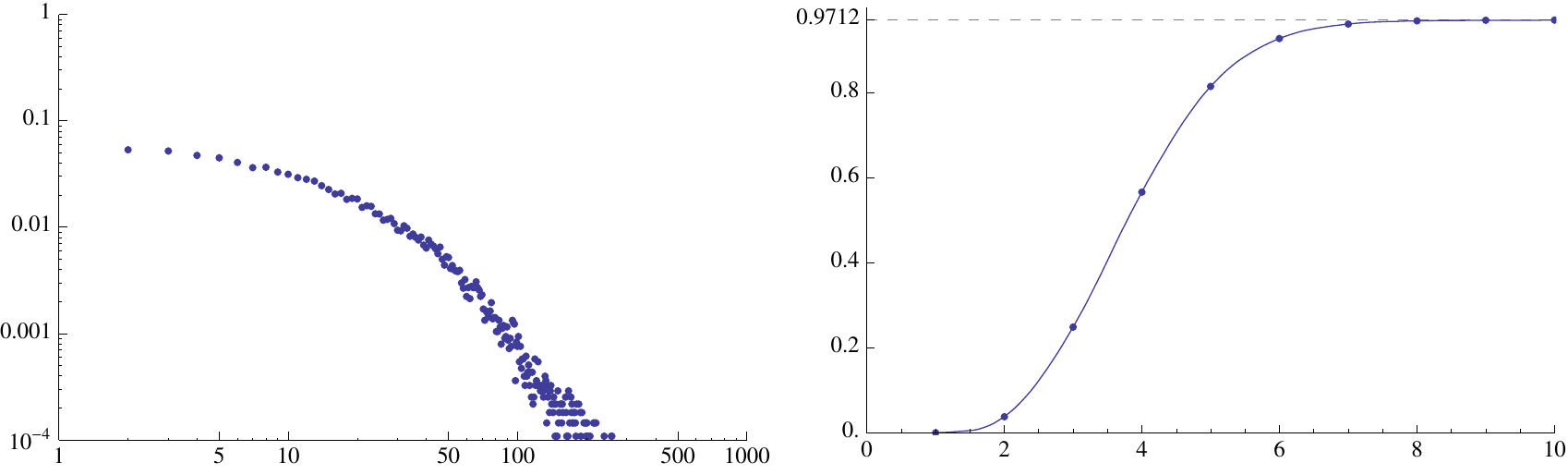}
\\
\hline
FF model &
\includegraphics[width=0.75\columnwidth]{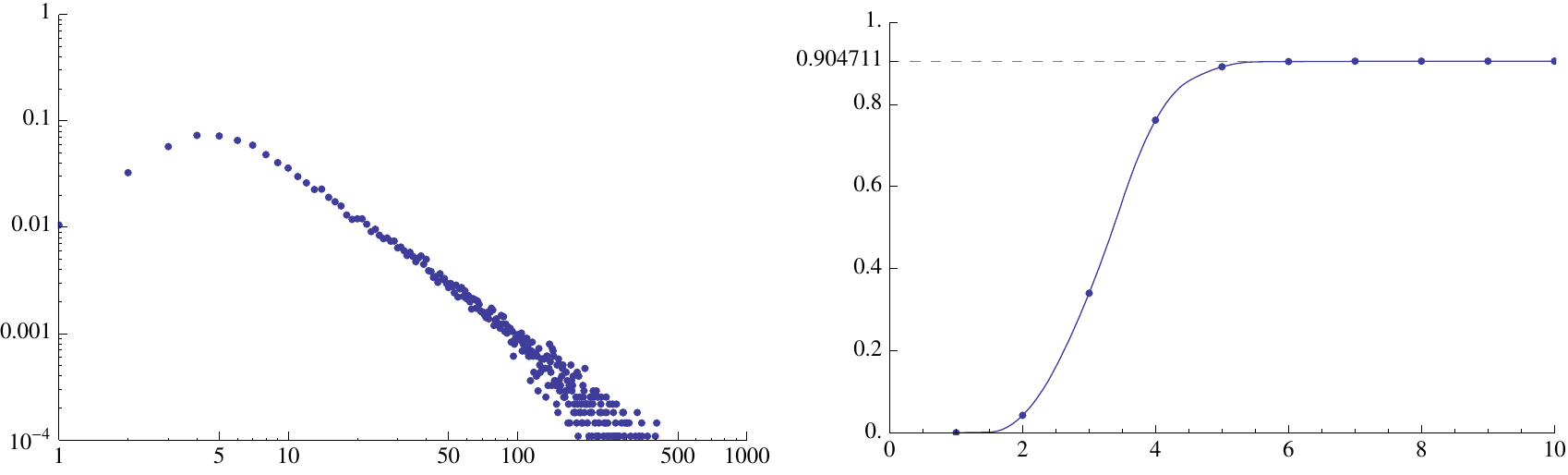}\\
\end{tabular} 
\caption{\label{fig:graph-analysis-real}
	Comparison of the \texttt{cit-HepTh} dataset
	\emph{versus} a graph generated by the FF model
	applied to the real feature matrix.
	We show the degree distribution in a log-log plot,
	and the fraction of pairs at distance at most $k$;
	in the latter, we highlight the peak value, indicating
	how many pairs of nodes are mutually reachable.
}
\end{center}
\end{figure}

\section{Conclusions}
\label{conclusions}

In this paper we introduce and study a network model that combines two
features:
\begin{enumerate}
 \item Behind the adjacency matrix of a network there is a {\em latent
   attribute structure} of the nodes, in the sense that each node is
   characterized by a number of features and the probability of the
   existence of an edge between two nodes depends on the features they
   share. 

\item Not all nodes are equally successful in transmitting their own
  attributes to the new nodes ({\em competition}). Each node $n$ is
  characterized by a random fitness parameter $R_n$ describing its
  ability to transmit the node's attributes: the greater the value of
  the random variable $R_n$, the greater the probability that a
  feature of $n$ will also be a feature of a new node, and so the
  greater the probability of the creation of an edge between $n$ and
  the new node.  Consequently, a node's connectivity does not depend
  on its age alone (so that also ``young'' nodes are able to compete
  and succeed in acquiring links).
 \end{enumerate}

Our work has different merits: firstly, we propose a simple model for
the latent bipartite ``node-attribute'' network, where the role played
by each single parameter is straightforward and easy to 
interpret: specifically, we have the two parameters, $\alpha$ and
$\beta$, that control the number of new attributes each new node
exhibits (in particular, $\beta>0$ tunes the power-law behavior of
the total number of distinct observed features); whereas the fitness
parameters $R_i$ impact on the probability of the new nodes to inherit
the attributes of the previous nodes.  Secondly, unlike other network
models based on the standard Indian Buffet Process, we take into
account the aspect of competition and, like in~\cite{bcpr-ibm}, we
introduce random fitness parameters so that nodes have a different
relevance in transmitting their features to the next nodes; finally,
we provide some theoretical, as well experimental, results regarding
the power-law behavior of the model and the estimation of the
parameters. By experimental data, we also show how the proposed model
for the attribute structure naturally leads to a complex network
model.

The comparison with real datasets is promising: our model seems to
produce quite realistic attribute matrices while at the same time
capturing most local and global properties (e.g., degree
distributions, connectivity and distance distributions) real networks
exhibit.

Some possible future developments are the following. First, we could
introduce another parameter $c\geq 0$ in the model of the
node-attribute bipartite network so that the inclusion probabilities
are
\begin{gather*}
P_n(k)=\frac{\sum_{i=1}^n R_i Z_{i,k}}{c+\sum_{i=1}^n R_i}
\end{gather*}
(we now have $c=0$): the bigger $c$, the smaller the inclusion
probabilities and so the sparser the attributes. This can allow to
obtain attribute matrices that are sparser. To this purpose, we note
that the proofs of the theoretical results regarding the estimation of
$\alpha$ and $\beta$ do not change and so, for this aspect, we have no
problem. The problem is, instead, in the fact that we have an
additional parameter to estimate.

Second, a possible variant of the feature/feature (FF) model is to
consider, for each incoming new node $i$, a feature/feature influence
matrix $\Xi(i)$ which depends on $i$: for instance, a diagonal matrix with 
$$
\xi_{k,k}(i)=\frac{1}{\sum_{\ell=1}^{i-1} Z_{\ell,k}}
 $$ so that the edge-probability is smaller as the number of nodes
with $k$ as a feature is larger.\\

\noindent {\bf Acknowledgments}

Paolo Boldi and Corrado Monti acknowledge the EU-FET grant NADINE (GA
288956).  They also would like to thank Andrea Marino for useful
discussions.  

Irene Crimaldi acknowledges support from CNR PNR Project
``CRISIS Lab''.  Moreover, she is a member of the Italian group
``Gruppo Nazionale per l'Analisi Matematica, la Probabilit\`a e le
loro Applicazioni (GNAMPA)'' of the Italian Institute ``Istituto
Nazionale di Alta Matematica (INdAM)''.

%\clearpage

%\bibliographystyle{elsarticle-num}

\bibliographystyle{plain}
\bibliography{mybib}

\end{document}